\documentclass[11pt]{article}
\usepackage{amsmath, amsthm}
\usepackage{amsmath,amssymb,latexsym}
\usepackage{graphicx}
\usepackage{epstopdf}
\newtheorem{theorem}{Theorem}
\newtheorem{lemma}{Lemma}

\usepackage{graphicx}
\usepackage{amsfonts}
\usepackage{subfigure}
\usepackage{appendix}
\newcommand{\argmax}{\operatornamewithlimits{argmax}}
\newcommand{\ddss}{\displaystyle}
\newcommand{\bbee}{\begin{equation}}
\newcommand{\eeee}{\end{equation}}
\newcommand{\bbaa}{\begin{array}}
\newcommand{\eeaa}{\end{array}}
\begin{document}
\thispagestyle{empty}
\setcounter{page}{1}
\setlength{\baselineskip}{1.5\baselineskip}
\title{Data Processing Bounds for Scalar Lossy Source Codes with Side Information at the Decoder \thanks{This research is supported by the Israeli Science Foundation (ISF),
grant no.\ 208/08.}}
\author{Avraham Reani and Neri Merhav\\
Department of Electrical Engineering\\
Technion - Israel Institute of Technology\\
Technion City, Haifa 32000, Israel\\
Emails: [avire@tx, merhav@ee].technion.ac.il}\maketitle
\begin{abstract}
In this paper, we introduce new lower bounds on the distortion of scalar fixed-rate codes for lossy compression with side information available at the receiver. These bounds are derived by presenting the relevant random variables as a Markov chain and applying generalized data processing inequalities a la Ziv and Zakai. We show that by replacing the logarithmic function with other functions, in the data processing theorem we formulate, we obtain new lower bounds on the distortion of scalar coding with side information at the decoder. The usefulness of these results is demonstrated for uniform sources and the convex function $Q(t)=t^{1-\alpha}$, $\alpha>1$. The bounds in this case are shown to be better than one can obtain from the Wyner-Ziv rate-distortion function.
%As far as we know, these bounds are the only existing non-trivial bounds for this situation.
\end {abstract}
\clearpage
\section{Introduction}
The Wyner--Ziv (WZ) problem has received very much attention during the
last three decades. 
There were several attempts to develop practical
schemes for lossy coding in the WZ setting, by using codes
with certain structures that facilitate the encoding and the
decoding. Most notably, these studies include nested
structures of linear coset codes (in the role of bins) for discrete
sources, and nested lattice structures for continuous valued
sources, see e.g., \cite{Kusuma01}, \cite{Servetto03}.
Other directions of introducing structure into WZ coding are
associated with trellis/turbo/LDPC designs (\cite{XLC04}
and references therein) and with progressive coding, i.e.,
successive refinement with layered code design \cite{SM04},
\cite{CX04}.
The case of scalar source codes for the WZ problem was also handled in several papers, e.g. \cite{kusuma} and \cite{muresan1}. Zero-delay coding strategies for the WZ problem, were introduced in \cite{Teneketzis}, where structure theorems for fixed-rate codes, under the assumption of a Markov source, were given. These results were later extended in \cite{kaspi}, to include variable-rate coding.
In \cite{tuncel1} and \cite{tuncel2} it was conjectured that under the high-resolution assumption, the optimal quantization level density is periodic. In addition, zero-delay schemes for specific source-side information correlation were presented in \cite{tuncel1}, \cite{tuncel2} and \cite{tuncel3}. Zero-delay coding of individual sequences under the conditions of the WZ problem was considered in \cite{reani}, where existence of universal schemes for fixed-rate and variable-rate coding was established.

In this paper, we develop lower bounds on the distortion in the scalar WZ setting. We apply the results of \cite{ZZ1} and \cite{ZZ2}, concerning functionals satisfying a data-processing
theorem, to this setting.
In \cite{ZZ1} it was shown that the rate-distortion (RD) bound ($R(D)\leq C$)
remains true when the negative logarithm function, in the definition of mutual information, is replaced
by an arbitrary convex, non--increasing function satisfying some technical conditions. For certain choices of this convex function, the bounds obtained were better than the classical RD bounds. These results were substantially generalized in \cite{ZZ2} to apply to even more general information measures. The methods of \cite{ZZ1} were also used in \cite{zamir1}, \cite{zamir2} and \cite{zamir3}. In these papers, lower bounds on the distortion of delay-constrained joint source-channel coding were given. These bounds were obtained by combining the R\'{e}nyi information measure \cite{Renyi} with the generalized data processing theorem of \cite{ZZ1}, and under high-resolution and high SNR approximations. Another related work is \cite{Merhav1}, where certain degrees of freedom of
the Ziv–-Zakai generalized mutual information were further exploited in order to get better bounds. It is worth to mention that there is some difference between our application of the generalized DPT and those of [15] and [16]. In [15] and [16], the higher term in the generalized RD bound (which will be referred as the generalized capacity) does not depend on the source distribution, whereas in our setting, it does. This is a direct consequence of the presence of side information which also has implications on the optimization problems we will encounter.

We start by presenting the relevant random variables of the WZ problem as a Markov chain. Then, using a data processing theorem, we obtain lower bounds on the distortion. We show that replacing the logarithmic function by other functions, may give better bounds on the distortion of delay--limited coding (in particular, for scalar coding) in the WZ setting. Examples of non-trivial lower bounds for scalar coding, in this setting, are obtained using the convex function $Q(t)=t^{1-\alpha}$, $\alpha>1$, which is equivalent to using the R\'{e}nyi information measure.
The importance of such bounds stems from the fact that finding the optimal scalar code in the WZ setting is, in general, a hard problem. In fact, it is a problem of finding an optimal partition of the source alphabet and this partition does not necessarily correspond to intervals. A main objective will be to use these bounds for studying the performance of concrete coding schemes. 

The remainder of the paper is organized as follows. In Section 2, we present our formulation of the WZ problem and establish a generalized data processing theorem (DPT) for this setting. In Subsection 2.1, we define the fixed-rate scalar coding case. We then give an upper bound on the generalized capacity, which is one component of the above DPT. 
In Subsection 2.2, we handle the second component of the generalized DPT, i.e., the generalized RD function. We start with a general characterization of this function. Then, we introduce a closed-form expression of the generalized RD function for uniformly distributed sources w.r.t. general symmetric distortion measures. In Section 3, we use the results of Section 2 to obtain non-trivial lower bounds on the distortion of scalar coding in the WZ setting in several cases. Finally, we demonstrate that for large alphabets, non-trivial bounds can be derived for various
channels and as a result, the performance range for scalar coding can be given.

%\subsection{Previous Work}
\section{Problem Formulation and Results}
In this section, we present the relevant random variables of the WZ problem as a Markov chain and establish a generalized data processing theorem (DPT) for this setting, using the method of \cite{ZZ1}.

We begin with notation conventions. Capital letters represent scalar random variables, specific realizations of them are denoted by the corresponding lower case letters and their alphabets - by calligraphic letters. The inner product of the two vectors $\vec{a}$ and $\vec{b}$ will be denoted by $\vec{a}\cdot\vec{b}$. Logarithms are defined to the base $2$.

We consider a memoryless source producing a random sequence $X_1,X_2,\ldots$ $X_i\in {\cal{X}}$, $i=1,2,\ldots$, where ${\cal{X}}$ is a finite alphabet with cardinality ${K}$. Without loss of generality, we define this alphabet to be the set $\{1,2,\ldots,{K}\}$. The probability mass function of $X$, $p(x)$, is known. A fixed-rate scalar source code with rate $R=\log M$,\footnotemark \footnotetext{Through this paper, the symbols of $Z$ are not necessarily transformed into bits. Therefore, $\log M$ need not necessarily be an integer.}
partitions ${\cal{X}}$ into $M$ disjoint subsets $( A_1, A_2,\ldots, A_M)$, $M\leq{K}$. The encoder maps $X_i$ into a channel symbol $Z_i$, using a function $f : {\cal{X}} \rightarrow$$\{1, 2, \ldots, M\}$, that is, $Z_i=f(X_i)$. The decoder, in addition to $Z_i$, has access to a random variable $Y_i$, which is dependent on $X_i$ via a known discrete memoryless channel (DMC), defined by the single-letter transition probability matrix $\left\{p(y|x)\right\}$, whose entries are the conditional probabilities of the different channel output symbols given the channel input symbols. Based on $Z_i$ and $Y_i$, the decoder produces the reconstruction $\hat{X}_i$, using a decoding function 
$g : \{1, 2, \ldots,M\}\times{\cal{X}}\rightarrow{\cal{X}}$, i.e., $\hat{X}_i=g(Z_i,Y_i)$.
This setting is depicted in Fig. \ref{WZ_Sett_fig}.
For simplicity, we assume that $X_i$, $Y_i$ and $\hat{X}_i$, all take on values in the same finite alphabet ${\cal{X}}$. The distortion in this setting is defined to be:
\begin{equation}
D=\mathbb{E} \rho(X_i,\hat{X}_i)=\ddss\sum_{x,y} p(x,y)\rho(x,\hat{x})
\eeee
where $p(x,y)$ is the joint distribution of $x$ and $y$ and $\rho(x,\hat{x})$ is a distortion measure, with $\rho(x,x)=0$.

Let $Q(t)$, $0 \leq t < \infty$, be a real-valued convex function, where 
\begin{equation}
\displaystyle\lim_{t\rightarrow 0} t\cdot Q(1/t)=0. 
\label{q_o_cond}
\end{equation}
This requirement implies that $Q(t)$ is non-increasing, as was shown in \cite{ZZ1}. We define $0\cdot Q(r/0)=0$, for all $0 \leq r < \infty$. 
The generalized mutual information relative to the function $Q$ is defined as
\begin{equation}
I^Q(X;Y)=\sum_{x,y} p(x,y) Q\left(\frac{p(y)}{p(y|x)} \right).
\end{equation}
\begin{figure}[!hc]
	\centering		
	\includegraphics[width=3.5in, height=1.5in]{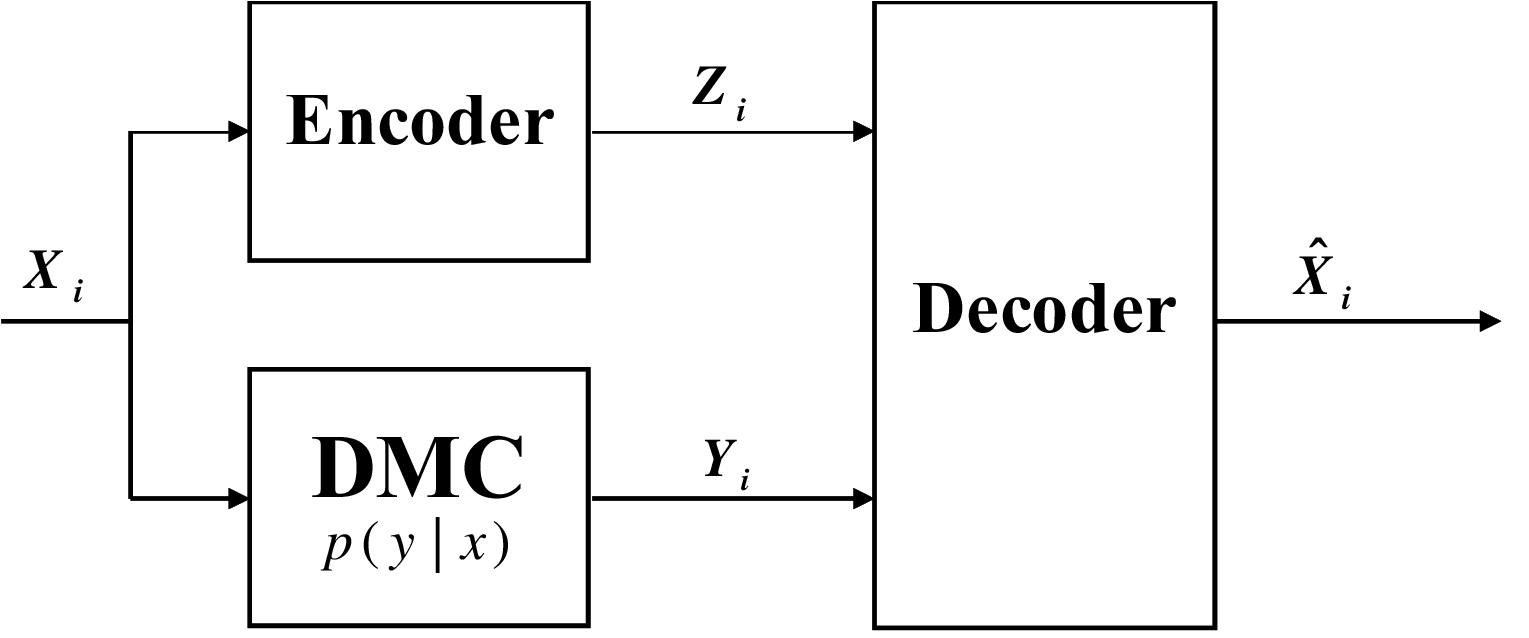}
	\caption { {The WZ setting}} 
	\label{WZ_Sett_fig}
\end{figure}
In the next steps, we will first allow $Z$ to be non-deterministic function of $X$. This will give us a more tractable form of the DPT.
We apply the generalized DPT \cite[Theorem 3]{ZZ1} in the following way:
\begin{equation}
I^Q(X;\hat{X})\leq I^Q(X;Y,Z)
\label{DPT1}
\end{equation}
where we have used the fact that $ X \leftrightarrow  (Y,Z) \leftrightarrow\hat{X}$ is a Markov chain.
Since $Z \leftrightarrow X \leftrightarrow Y$ is also a Markov chain, we have:
\begin{equation}
p(x,y,z)=p(x)p(y|x)p(z|x)
\end{equation}
and $I^Q(X;Y,Z)$ is given by:
\begin{eqnarray}
\nonumber I^Q(X;Y,Z)&=&\ddss\sum_{x,y,z}p(x)p(y|x)p(z|x)Q\left(\ddss\frac{p(y)p(z|y)}{p(y,z|x)}\right)\\
\nonumber &&\\
\nonumber &=&\ddss\sum_{x,y,z}p(x)p(y|x)\ddss p(z|x)Q\left(\ddss\frac{p(y)p(z|y)}{p(y|x)p(z|x)}\right)\\
\nonumber &&\\
\nonumber &=&\ddss\sum_{x,y,z}p(x)p(y|x)\ddss p(z|x)Q\left(\ddss\frac{\ddss\sum_{\tilde{x}}p(\tilde{x})p(y|\tilde{x})p(z|\tilde{x})}{p(y|x)p(z|x)}\right)\\
\nonumber &&\\
&=&\ddss\sum_{y}\ddss\sum_{z}\ddss\sum_{x}p(x)p(y|x)p(z|x)Q\left(\ddss\frac{\ddss\vec{p}_z\cdot \vec{p}_{y}}{p(y|x)p(z|x)}\right)\\
\nonumber\label{mutualinfo}
\end{eqnarray}
where we have defined the following ${K}$-dimensional vectors $\displaystyle\left\{ \vec{p}_z \right\}_{z=1}^{M}$:
\begin{eqnarray}
\vec{p}_z&=&[p(z|x), x\in {{\cal{X}}}]
\label{pznddef}
\end{eqnarray}
and the following ${K}$-dimensional vectors $\displaystyle\left\{ \vec{p}_{y}\right\}$, $y\in {{\cal{X}}}$:
\begin{eqnarray}
\vec{p}_{y}&=&[p(x,y),x\in {{\cal{X}}}].
\label{pynnddef}
\end{eqnarray}
%with the understanding that the source sequences in ${{\cal{X}}}^n$ are numbered according to some order, say, the lexicographic order.
By definition of $\displaystyle\left\{ \vec{p}_z \right\}_{z=1}^{M}$, we have the following property:
\begin{equation}
\sum_{z=1}^M \vec{p}_z = [1,1,\ldots,1].
\label{cons1}
\end{equation}
We now define the following functions $\left\{G_{y}(\vec{p}_z)\right\}_{y=1}^K$, $G_y: {\mathbb{R}}_+^{K} \rightarrow \mathbb{R}$:
\begin{equation}
G_{y}(\vec{p}_z)=\ddss\sum_{{x}}p({x})p({y}|{x})p(z|x)Q\left(\ddss\frac{\ddss\vec{p}_z\cdot \vec{p}_{{y}}}{p({y}|{x})p(z|x)}\right).
\label{impGfunctiondef}
\end{equation}
Using these functions, Eq. (\ref{mutualinfo}) becomes:
\begin{eqnarray}
I^Q(X;Y,Z)&=&\ddss\sum_{y}\ddss\sum_{z} G_{y}(\vec{p}_z).
\label{importantI}
\end{eqnarray}
The functions $G_{y}(\vec{p}_z)$ have the following property:
\begin{lemma}
For any convex function $Q$ satisfying (\ref{q_o_cond}), the functions $\left\{G_{y}(\vec{p}_z)\right\}$, $y\in {{\cal{X}}}$, are convex.
\label{convexityGy}
\end{lemma}
The proof is given in Appendix A. This convexity property has important implications in the optimization of $I^Q(X;Y,Z)$, as will be discussed later.
Assuming the encoder is given by a deterministic function $f: {{\cal{X}}} \rightarrow \{1,\ldots, M\}$, Eq. (\ref{mutualinfo}) becomes the following:
\begin{eqnarray}
\nonumber I^Q(X;Y,Z) 
\nonumber &=&\ddss\sum_{y}\ddss\sum_{x}p(x)p(y|x)\ddss\sum_{z}p(z|x)Q\left(\ddss\frac{\ddss\vec{p}_z\cdot \vec{p}_{y}}{p(y|x)p(z|x)}\right)\\ 
\nonumber &=&\ddss\sum_{y}\ddss\sum_{x}p(x)p(y|x)\ddss Q\left(\ddss\frac{\ddss\vec{p}_{f(x)}\cdot \vec{p}_{y}}{p(y|x)}\right)\\ 
&=&\displaystyle\sum_{x,y}p(x,y)Q\left(\displaystyle\frac{\displaystyle\sum_{\tilde{x}\in A_z} p(\tilde{x},y)} {p(y|x)} \right)\\
\nonumber\label{base_eq1}
\end{eqnarray}
where $z=f(x)$ and $A_z\equiv\{\tilde{x}: f(\tilde{x})=z\}$. Remember that we have defined $0\cdot Q(r/0)=0$.
Using $Q(t)=-\log t$ in (\ref{base_eq1}), thus turning back to the classical DPT, we next show the following result:
%\begin{theorem}
\begin{equation}
R(D)-I(X;Y)\leq \sup  H(Z|Y)
\label{wzlbtheo22}
\end{equation}
%\label{wzlbtheo}
%\end {theorem}
where $D=\mathbb{E} \rho(X,\hat{X})$, $R(D)$ is the classical RD function and the supremum is taken over all $Z=f(X)$, where $f : {\cal{X}} \rightarrow$$\{1, 2, \ldots, M\}$. 
This inequality stems from the Markov properties of the WZ problem we discussed before. The optimization over all scalar encoders, here and in the following, gives us lower bounds on the distortion of $any$ such encoder. We see that given a rate $R=\log M$, we should find the encoder that maximizes $H(Z|Y)$. This is not surprising as, intuitively, we want the amount of information that $Y$ has on $Z$ to be as little as possible, to decrease the redundancy. Ideally, we want the encoder output and the side information to be independent. This is indeed achieved by the block coding scheme of Wyner and Ziv, in the limit of infinite block length. The term $\sup \left\{ H(Z|Y)\right\}+I(X;Y)$ will be referred to as the ``capacity'' of the generalized channel between $(X,Z)$ and $Y$. This channel is composed of the DMC between $X$ and $Y$ and a noiseless channel with capacity $\log M$ for the encoder's output. Since the source distribution is given, the maximum rate of reliable communication over this channel is indeed $I(X;Y)+\sup H(Z|Y)$.
\begin{proof}[Proof of Eq. (\ref{wzlbtheo22})]
Using the function $Q(t)=-\log t $ in (\ref{base_eq1}), we get:
\begin{eqnarray}
I^Q(X;Y,Z)
%&=& I(X;Y)-\displaystyle \sum_{x,y}p(x,y)\log\left(\displaystyle\sum_{x'\in A_z} p(x'|y)\right)\\
\nonumber &=& H(Y,Z)-H(Y,Z|X)\\%I(X;Y)-\displaystyle \sum_{x,y}p(x,y)\log\displaystyle p(z|y)\\
\nonumber &=& H(Y)+H(Z|Y)-H(Y|X)-H(Z|Y,X)\\%I(X;Y)-\displaystyle \sum_{y}\sum_{z}\sum_{x\in A_z}p(x,y)\log  p(z|y)\\
&=& I(X;Y)+H(Z|Y)
\end{eqnarray}
where we have used the fact that $H(Z|Y,X)=H(Z|X)=0$ since $Z$ is a deterministic function of $X$. On substituting into (\ref{DPT1}), we get:
\begin{eqnarray}
\nonumber R(D) &\leq& I^Q(X;\hat{X})\\
\nonumber &\leq& I^Q(X;Y,Z)\\
\nonumber &=&I(X;Y) +H(Z|Y)\\
&\leq& I(X;Y) + \sup  H(Z|Y),
\end{eqnarray}
which is equivalent to Eq. (\ref{wzlbtheo22}).
%\begin{equation}
%n\cdot \left[ R(D)-I(X;Y)\right] \leq H(Z|Y^n) \leq  \sup \left\{ H(Z|Y^n)\right\}
%\label{classicalineq}
%\end{equation}
\end{proof}
Notice that if we allow non-deterministic encoders, as in Eq. (\ref{mutualinfo}), we get the following:
\begin{equation}
R(D)\leq I(X;Y) + \sup \left\{ H(Z|Y)-H(Z|X)\right\},
\end{equation}
where the supremum is taken over the same set as in (\ref{wzlbtheo22}). Although randomizing the encoder can increase $H(Z|Y)$, it will also increase $H(Z|X)$. Due to the convexity property presented in Lemma \ref{convexityGy}, the supremum is achieved by a deterministic encoder, as will be discussed in the next section. Therefore, randomizing the encoder cannot improve the bound in this setting.

In Section 3, we show some examples of scalar coding, where this result gives us lower bounds on the distortion, which are better than the bounds obtained from the classical inequality $R_{WZ}(D)\leq \log M$, where $R_{WZ}(D)$ is the WZ RD function.
\subsection{Generalized DPT for fixed-rate scalar coding}
Assuming a deterministic encoder, the vectors $\displaystyle\left\{ \vec{p}_z \right\}_{z=1}^{M}$, defined in (\ref{pznddef}), become:
\begin{equation}
\vec{p}_z=[1_{1\in A_z},1_{2\in A_z}, \ldots, 1_{{{K}}\in A_z}],
\label{pzdef}
\end{equation}
where $1_B$ is the indicator function of the event $B$. The $j$th coordinate of $\vec{p}_z$ is $1$ if $j\in A_z$ and $0$ elsewhere. 
Using these vectors, we can rewrite (\ref{base_eq1}) in the following way:
\begin{eqnarray}
\nonumber I^Q(X;Y,Z) %&=&\displaystyle\sum_{x,y}p(x,y)Q\left(\frac{p(y)\displaystyle\sum_{\tilde{x}\in A_{z(x)}} p(\tilde{x}|y)}{p(y|x)}\right)\\
\nonumber &=&\displaystyle\sum_{x,y}p(x,y)Q\left(\frac{\displaystyle \vec{p}_{z(x)} \cdot \vec{p}_{y}}{p(y|x)}\right)\\
\nonumber &=&\displaystyle\sum_y \sum_{z}\sum_{x\in A_{z}}p(x,y)Q\left(\displaystyle\frac{\displaystyle 
\vec{p}_{z} \cdot \vec{p}_{y}}{p(y|x)}\right)\\
\nonumber &=&\displaystyle\sum_y \sum_{z} \vec{p}_{z} \cdot \vec{q}_{{z},y}\\
&=& \displaystyle\sum_y \sum_{z} \Gamma_y(\vec{p}_{z}),
\label{GMI11}
\end{eqnarray}
where we have defined the following ${K}$-dimensional vectors:
\begin{equation}
\begin{array}{l}
\vec{q}_{z,y}=\left[\ddss p(x_1,y)Q\left(\frac{\vec{p}_z \cdot \vec{p}_{y}}{p(y|x_1)}\right),
\displaystyle p(x_2,y)Q\left(\frac{\displaystyle \vec{p}_z \cdot \vec{p}_{y}}{p(y|x_2)}\right),\ldots\right]\\
\end{array}
\end{equation}
and the set of functions $\left\{\Gamma_y \right\}_{y=1}^{{K}}$, $\Gamma_y: \mathbb{R}_+^{{K}} \rightarrow \mathbb{R}$:
\begin{equation}
\Gamma_y(\vec{p}_{z})=\vec{p}_{z} \cdot \vec{q}_{{z},y}.
\label{funcG}
\end{equation}
Notice that the vector $\vec{p}_{y}$ depends only on $y$ and that the inner product  $\vec{p}_z \cdot \vec{p}_{y}$ is a function of $z$ and $y$.
Applying the RD bound \cite[Theorem 4]{ZZ1}, we get:
\begin{equation}
R^Q(D)\leq I^Q(X;\hat{X})\leq I^Q(X;Y,Z)\leq C^Q,
\label{DPT2}
\end{equation}
where
\begin{eqnarray}
R^Q(D)&=&\inf I^Q(X;\hat{X})
\end{eqnarray}
and
\begin{eqnarray}
\nonumber C^Q&=&\sup I^Q(X;Y,Z)\\
&=&\sup \displaystyle\sum_y p(y)\sum_{z} \Gamma_y(\vec{p}_{z}).
\label{infsup}
\end{eqnarray}
This gives us the following lower bound on the distortion $D$:
\begin{equation}
D\geq  D^Q(C^Q),
\end{equation}
where $D^Q(R)$ is the inverse function of $R^Q(D)$.
The infimum is taken over all conditional distributions $\{p(\hat{x}|x)\}$ that satisfy the distortion constraint $\mathbb{E} \rho(X,\hat{X}) \leq D$. The supremum should be taken over all scalar encoders with a fixed rate $R=\log M$. Alternatively, we can carry out a continuous optimization by taking the supremum over all sets of positive vectors $\left\{ \vec{p}_z\right\}_{z=1}^M$ that satisfy (\ref{cons1}), i.e., all conditional distributions $\{p(z|x)\}$. Whereas the original optimization problem may require exhaustive search over all encoders, and in this case, our mechanism is useless, the continuous problem may have analytic solution. The result of the continuous optimization will, of course, be greater than or equal to $C^Q$. However, the functions $\{\Gamma_y(\vec{p}_{z})\}$ might be neither convex nor concave. In this case, we can carry out the optimization using the general form of $I^Q(X;Y,Z)$ given in (\ref{importantI}), which is convex in $\vec{p}_z$.

Until now, we only handled fixed-rate codes. However, distortion lower bounds for codes created by time-sharing fixed-rate codes are readily obtained from the above. This can be seen as follows: For a given rate $R\in\mathbf{R} \triangleq \{\log 1, \log2,\ldots, \log K\}$, let $D(R)$ be the minimum distortion achievable by fixed-rate scalar codes with encoders $f: {\cal{X}}\rightarrow \{1,2,\ldots,2^R\}$ and let $\underline{D}(R)$ be a lower bound on this distortion. We construct a variable-rate code whose rate at time $t$ is $R_t$, $R_t\in \mathbf{R}$, $t=1,\ldots,n$, by time-sharing scalar fixed-rate codes, under the constraint:
\begin{eqnarray}
\ddss\frac{1}{n}\ddss\sum_{t=1}^n R_t \leq R.
\end{eqnarray}
The distortion of this time-sharing code is lower bounded by:
\begin{eqnarray}
\nonumber D&\geq& \ddss\frac{1}{n}\ddss\sum_{t=1}^{n} D(R_t)\\
\nonumber &\geq&\ddss\frac{1}{n}\ddss\sum_{t=1}^{n} \underline{D}(R_t)\\
\nonumber &\geq&\ddss\frac{1}{n}\ddss\sum_{t=1}^{n}\underline{D}^*(R_t)\\
\nonumber&\geq&  \underline{D}^*\left(\ddss\frac{1}{n}\ddss\sum_{t=1}^{n}R_t\right)\\
&\geq&  \underline{D}^*\left(R\right),
\end{eqnarray}
%where $\{\underline{D}(R_i)\}_{i=1}^{K}$ is a set of distortion lower bounds for the fixed-rate codes
where $\underline{D}^*(R)$ is the lower convex envelope of the set $\{\underline{D}(R)\}_{R\in \mathbf{R}}$ and is defined by:
\begin{eqnarray}
\underline{D}^*(R)\triangleq \min\ddss\sum_{i=1}^K \beta_i \underline{D}(\log i )
\end{eqnarray}
where the minimum is taken over the following set:
\bbee
\bbaa{lllll}
\{\beta_1,\beta_2,\ldots,\beta_K: &\beta_i \geq 0& \forall i,&\ddss\sum_{i=1}^K \beta_i=1,&\ddss\sum_{i=1}^K \beta_i\log i \leq R\}
\eeaa
\eeee
We see that $\underline{D}^*(R)$ lower bounds the distortion of any such time-sharing code with rate no more than $R$. Concrete examples of this result will be given in the next section.

We end this subsection with an upper bound on $C^Q$ for the specific convex function $Q(t)=t^{1-\alpha}$, $\alpha>1$. Using this choice of $Q$ is equivalent to using the R\'{e}nyi mutual information of order $\alpha$, which is defined as \cite{Renyi}:
\begin{eqnarray}
\nonumber I^r_{\alpha}(X;Y)&\equiv&\displaystyle\frac{1}{\alpha-1}\log \displaystyle \sum_{x,y}p(x,y)\displaystyle\left[\frac{p(x)}{p(x|y)}\right]^{1-\alpha}\\
&=&\displaystyle\frac{1}{\alpha-1}\log \displaystyle I^Q(X;Y).
\end{eqnarray}
Thus, Eq. (\ref{DPT2}) can be written in the following equivalent form:
\begin{equation}
R^r_{\alpha}(d)\leq I^r_{\alpha}(X;\hat{X})\leq I^r_{\alpha}(X;Y,Z)\leq C^r_{\alpha},
\end{equation}
where
\begin{eqnarray}
R^r_{\alpha}(d)&=&\displaystyle\frac{1}{\alpha-1}\log \displaystyle R^Q(D)
\end{eqnarray}
and
\begin{eqnarray}
C^r_{\alpha}&=&\displaystyle\frac{1}{\alpha-1}\log \displaystyle C^Q.
\end{eqnarray}
The logarithmic measure is a special case, obtained for $\alpha\rightarrow 1$. Thus, optimizing over $\alpha$ can only improve the classical bounds. In addition, the function $Q(t)=t^{1-\alpha}$ is relatively convenient to work with.
\begin{theorem}
For the convex function $Q(t)=t^{1-\alpha}$, $1 < \alpha < 2$, we have the following upper bound:
\begin{equation}
C^Q\leq \displaystyle M^{\alpha-1}\sum_y \left(\sum_{x}p(x)\cdot p(y|x)^{\frac{1}{2-\alpha}} \right)^{2-\alpha}.
\label{holder1b}
\end{equation}
\label{holder1btheorem}
\end{theorem}
Notice that upper-bounding $C^Q$ will, of course, decrease the lower bounds on the distortion.
\begin{proof}
Using the function $Q(t)=t^{1-\alpha}$ in (\ref{GMI11}), we get:
\begin{eqnarray}
\nonumber I^Q(X;Y,Z) &=&\displaystyle\sum_y\sum_z\sum_{x\in A_z}p(x,y)
\left( \displaystyle\frac{\displaystyle \vec{p}_z \cdot \vec{p}_{y}}{p(y|x)}\right)^{1-\alpha}\\
\nonumber &=&\displaystyle\sum_y \sum_z\sum_{x\in A_z}
 p(x)p(y|x)^{\alpha}\cdot \displaystyle [\vec{p}_z \cdot \vec{p}_{y}]^{1-\alpha}\\
\nonumber &=&\displaystyle\sum_y \sum_z\displaystyle [\vec{p}_z \cdot \vec{p}_{y}]^{1-\alpha}\sum_{x\in A_z}
 p(x)p(y|x)^{\alpha}\\
&=&\displaystyle\sum_y\sum_z \displaystyle\left[\ddss\sum_x p(x,y)1_{x\in A_z}\right]^{1-\alpha}\left[\ddss\sum_x p(x) p(y|x)^{\alpha}1_{x\in A_z}\right].
\label{conv22}
\end{eqnarray}
In order to get an upper bound on $I^Q(X;Y,Z)$, we define:
\begin{equation}
\begin{array}{lll}
q&=&\displaystyle{1}/{(\alpha-1)}\\
r&=&\displaystyle {1}/{(2-\alpha)}\\
\end{array}
\end{equation}
and the following ${K}$-dimensional vectors:
\begin{equation}
\begin{array}{lll}
\vec{a}_y&=&\left[p(x_1,y)^{\alpha-1}\cdot 1_{1\in A_z}, p(x_2,y)^{\alpha-1}\cdot 1_{2\in A_z}, \ldots\right],\\
&&\\
\vec{b}_y&=&\left[p(x_1)^{2-\alpha}\cdot p(y|x_1)\cdot 1_{1\in A_z}, p(x_2)^{2-\alpha}\cdot p(y|x_2)\cdot 1_{2\in A_z}, \ldots\right].\\
\end{array}
\end{equation}
Applying these definitions to (\ref{conv22}), we have:
\begin{eqnarray}
I^Q(X;Y,Z)&=&\displaystyle\sum_y  \sum_z \left(\displaystyle\sum_{k=1}^{{K}} a_{y,k}^q\right)^{-1/q}\cdot\displaystyle \left(\vec{a}_y\cdot \vec{b}_y\right).
\end{eqnarray}
Assuming $\alpha$ is in the range $1<\alpha<2$, we have  $1 \leq q,r < \infty$ and $1/q+1/r=1$. Thus we can apply the H\"{o}lder inequality to each term in the sum, in the following way:
\begin{equation}
\displaystyle \left(\vec{a}_y\cdot \vec{b}_y\right)\cdot \left(\displaystyle\sum_{k=1}^{{K}} a_{y,k}^q\right)^{-1/q} \leq \left(\displaystyle\sum_{k=1}^{{K}} b_{y,k}^r\right)^{1/r} 
\end{equation}
We then have:
\begin{eqnarray}
\nonumber I^Q(X;Y,Z)
&\leq&\displaystyle\sum_y  \sum_z \left(\displaystyle\sum_{k=1}^{{K}} b_{y,k}^r\right)^{1/r} \\
\nonumber &=&\displaystyle\sum_y  \sum_z \left(\displaystyle \vec{p}_z\cdot [p(x_1)\cdot p(y|x_1)^{\frac{1}{2-\alpha}}, p(x_2)\cdot p(y|x_2)^{\frac{1}{2-\alpha}}, \ldots]      \right)^{2-\alpha}\\
\nonumber &=&\displaystyle\sum_y   \sum_z \displaystyle M\cdot \frac{1}{M}\left(\displaystyle \vec{p}_z\cdot [p(x_1)\cdot p(y|x_1)^{\frac{1}{2-\alpha}}, p(x_2)\cdot p(y|x_2)^{\frac{1}{2-\alpha}}, \ldots]        \right)^{2-\alpha}\\
\nonumber &\leq&\displaystyle\sum_y    \displaystyle M\cdot \left(\displaystyle \frac{1}{M}\sum_z \vec{p}_z\cdot [p(x_1)\cdot p(y|x_1)^{\frac{1}{2-\alpha}}, p(x_2)\cdot p(y|x_2)^{\frac{1}{2-\alpha}}, \ldots]       \right)^{2-\alpha}\\
&=&M^{\alpha-1}\displaystyle\sum_y    \displaystyle  \left(\displaystyle \sum_{x} p(x)\cdot p(y|x)^{\frac{1}{2-\alpha}} \right)^{2-\alpha},
\label{holder_bound}
\end{eqnarray}
where the second inequality is due to Jensen, using the fact that the function $q(t)=t^{2-\alpha}$ is concave for $1<\alpha<2$. The last equality follows from the constraint (\ref{cons1}).
\end{proof}
The usefulness of this result stems from its generality. It holds for any source distribution and any transition probability matrix $\{p(y|x)\}$. This result is used in Section 3, along with tighter bounds on the capacity that can be achieved in several special cases. It will also be shown that the application of this result to large alphabets yields non-trivial bounds. For $Q(t)=-\log t$, Eq. (\ref{holder1b}) is equivalent to the following:
\begin{eqnarray}
\nonumber C&\leq&\displaystyle\lim_{\alpha\rightarrow 1}\displaystyle\frac{1}{\alpha-1}\log\left(M^{\alpha-1}\displaystyle\sum_y    \displaystyle  \left(\displaystyle \sum_{x} p(x)\cdot p(y|x)^{\frac{1}{2-\alpha}} \right)^{2-\alpha}\right)\\
&=&\log M+I(X;Y),
\label{boundclassicholder}
\end{eqnarray}
where $C$ is the classical capacity.
Therefore, Eq. (\ref{holder1b}) can be viewed as generalization of (\ref{boundclassicholder}). Notice that the bound in (\ref{boundclassicholder}) can be derived easily from (\ref{wzlbtheo22}). This simple bound states that a maximum amount of information is transferred to the decoder when the output of the deterministic encoder is uniformly distributed and independent of the side information.
The proof of (\ref{boundclassicholder}) is given in Appendix B.
%%%%%%%%%%%%%%%%%%%%%%%%%%%%%%%%%%%%%%%%%%%%%%%%%%%%%%%%%%%%%%%%%%%%%%%%%%%%%%%%%%%%%%%%%%%%%%%%%%%%%%%%%%%%%%%%%%%%%%%%%%%%%%
\subsection{The generalized rate-distortion function for the uniform source distribution}
In this subsection, we handle the left-hand side of (\ref{DPT2}), i.e., the generalized RD function. We start with a general characterization. Then, in Lemma \ref{RDgenerallemma}, we give a closed-form expression for the generalized RD function of uniformly distributed sources w.r.t. general symmetric distortion measures. Finally, in Lemma \ref{RDHamming123}, we provide an explicit expression of this function for the special case of the Hamming distortion measure. These results will be used in the next section to derive concrete lower bounds on the distortion, from the DPT we formulated in (\ref{DPT2}).

By definition of the generalized mutual information:
\begin{eqnarray}
\nonumber I^Q(X;\hat{X})&=&\displaystyle\sum_{x,\hat{x}}p(x,\hat{x})Q\left(\frac{p(\hat{x})}{p(\hat{x}|x)}\right)\\
\nonumber &=&\displaystyle\sum_{x,\hat{x}}p(x)p(\hat{x}|x)Q\left(\frac{\ddss\sum_{x'}{p(x')p(\hat{x}|x')}}{p(\hat{x}|x)}\right)\\
\nonumber &=&\displaystyle\sum_{\hat{x}}\sum_{x} p(x)p(\hat{x}|x)Q\left(\frac{\vec{p}\cdot\vec{p}_{\hat{x}}}{p(\hat{x}|x)}\right)\\
&=&\displaystyle\sum_{\hat{x}}\Psi(\vec{p}_{\hat{x}}),
\label{proofrd1}
\end{eqnarray}
where we have defined the following $K$-dimensional vectors:
\bbee
\bbaa{lll}
\vec{p}&=&[p(x), x\in {\cal{X}}]\\
\vec{p}_{\hat{x}}&=&[p(\hat{x}|x), x\in {\cal{X}}]\\
\eeaa
\eeee
and the following function:
\bbee
\Psi(\vec{p}_{\hat{x}})=\ddss\sum_{x} p(x)p(\hat{x}|x)Q\left(\frac{\vec{p}\cdot\vec{p}_{\hat{x}}}{p(\hat{x}|x)}\right).
\eeee
For any convex function $Q$, $\Psi(\vec{p}_{\hat{x}})$ is a convex function. This can be shown easily by the same method we used to prove Lemma \ref{convexityGy}. The generalized RD function is given by:
\bbee
R^Q(D)=\inf \left\{ \displaystyle\sum_{\hat{x}}\Psi(\vec{p}_{\hat{x}})\right\},
\eeee
where the infimum is taken over all conditional distributions $\{p(\hat{x}|x)\}$ under the constraint:
\begin{eqnarray}
\ddss\sum_{\hat{x}}\sum_{x} p(x)p(\hat{x}|x)\rho(x,\hat{x})\leq D.
\end{eqnarray}
Notice that $R^Q(D)$ is lower bounded by $Q(1)$. This can be shown easily by the following:
\begin{eqnarray}
\nonumber I^Q(X;\hat{X})&=&\displaystyle\sum_{x,\hat{x}}p(x,\hat{x})Q\left(\frac{p(\hat{x})}{p(\hat{x}|x)}\right)\\
\nonumber&\geq&\displaystyle Q\left(\ddss\sum_{x,\hat{x}}\frac{p(x,\hat{x})p(\hat{x})}{p(\hat{x}|x)}\right)\\
&=&Q(1)
\label{q_1_proof}
\end{eqnarray}
where the inequality is due to Jensen.
$I^Q(X;\hat{X})$ is, of course, convex in the set $\left\{ \vec{p}_{\hat{x}}\right\}_{\hat{x}=1}^{K}$. Thus, this is a standard problem of minimizing a convex function over a convex set under linear constraints. Generally, this optimization problem can be solved numerically by various algorithms (see, e.g., \cite[Chap. 3]{ConvexOptimization}). 

In the next steps we will give analytic expressions to the generalized RD function under certain conditions.
We refer to a distortion measure $\rho(x,\hat{x})$ as symmetric if the rows of the distortion matrix, $\left\{\rho(x,\hat{x})\right\}$, are permutations of each other and the columns are permutations of each other.
A uniformly distributed source is a source for which $p(x)=\displaystyle\frac{1}{{K}}$, $\forall x\in {\cal{X}}$.
\begin{lemma}
Consider a discrete source $X$, uniformly distributed over a finite alphabet ${{\cal{X}}}$, and let $Q(t)$, $0 \leq t < \infty$, be any real-valued differentiable convex function satisfying (\ref{q_o_cond}). Then, $R^Q(D)$ w.r.t. any symmetric distortion measure is given by:
\begin{equation}
R^Q(D)=\ddss\left\{\begin{array}{ll}\displaystyle\sum_{k=1}^{K} p_k\cdot Q\left(\frac{1}{K p_k}\right),& D\leq \ddss\frac{1}{K}\sum_{k=1}^K\rho_k\\
Q(1),& D>\ddss\frac{1}{K}\sum_{k=1}^K\rho_k
\end{array}\right.
\end{equation}
where $\ddss\left\{\rho_k\right\}_{k=1}^{K}$ are the elements of each row of the matrix $\left\{\rho(x,\hat{x})\right\}$ and $\ddss\left\{ p_k\right\}_{k=1}^{K}$ is a probability distribution, which is given by the following equations ($k=1,\ldots,K$):
\begin{equation}
\ddss Q\left(\frac{1}{K p_k}\right)-\displaystyle \frac{1}{K p_k} Q'\left(\frac{1}{K p_k}\right)+\lambda_1+\lambda_2 \rho_k-\mu_k=0, 
\label{KKT1}
\end{equation}
where $\lambda_1$, $\lambda_2$, $\ddss\left\{ \mu_k\right\}_{k=1}^{K}$ are constants, chosen such that:
\begin{equation}
\begin{array}{llll}
\ddss\sum_{k=1}^{K}p_k=1,&\ddss\sum_{k=1}^{K}p_k \rho_k = D&\\
%\ddss\sum_{k=1}^{K}p_k \rho_k &=& d&\\
\mu_k\geq 0,& \mu_k\cdot p_k=0 ,& k=1,\ldots,K
\end{array}
\label{KKT2}
\end{equation}
\label{RDgenerallemma}
\end{lemma}
Notice that the equations (\ref{KKT1}) are decoupled, thus each $p_k$ can be calculated separately. The proof of Lemma \ref{RDgenerallemma} is given in Appendix C.
\\
$\mathbf{Example.}$ Taking $Q(t)=t^{-s}$, we get:
\begin{equation}
 (s+1)p^s_k+\lambda_1+\lambda_2 \rho_k-\mu_k=0,
\end{equation}
which is equivalent to the following:
\begin{equation}
p_k=c(\mu_k+1-\lambda\rho_k)^{\frac{1}{s}},
\end{equation}
where specific value of $\lambda$ matches to a point on the generalized RD curve and $c$ is a normalization factor.

For the Hamming distortion measure, defined by:
\begin{equation}
\rho(x,\hat{x})=\left\{ \begin{array}{ll} 0 & x=\hat{x} \\ 1 & x\neq \hat{x}        \end{array} \right.
\label{HammDis}
\end{equation}
we have the following closed-form expression:
\begin{lemma}
Consider a discrete source $X$, uniformly distributed over a finite alphabet ${{\cal{X}}}$, and let $Q(t)$, $0 \leq t < \infty$ be any real-valued convex function satisfying (\ref{q_o_cond}). Then, $R^Q(D)$ w.r.t. the Hamming distortion measure is given by:
\begin{equation}
R^Q(D)=\left\{\begin{array}{ll}\displaystyle (1-D) \cdot Q\left[\frac{1}{K(1-D)}         \right]+\displaystyle D\cdot Q\left[   \frac{   {K}-1 }{{KD}}\right],
& D\leq \frac{K-1}{K}\\
Q(1),& D> \frac{K-1}{K}
\end{array}\right.
\label{RDF}
\end{equation}
\label{RDHamming123}
\end{lemma}
Notice that Lemma \ref{RDHamming123} does not require the differentiability of the convex function $Q$.
The proof is given in Appendix D.

The general form of $R^Q(D)$ enables the use of any convex function $Q$. These results make the Ziv-Zakai mechanism much more tractable, at least for the case of uniform sources. In addition, they provide direct solutions for a broad class of classical RD functions.  We use these results in the next section to derive non-trivial bounds on the distortion of scalar coding in several cases.
\section {Applications}
In this section, we use the results of the Section 2 to derive lower bounds on the distortion in several cases. Non-trivial bounds are obtained using the convex function $Q(t)=t^{1-\alpha}$, $\alpha>1$, which was mentioned above.
We assume that the source is uniformly distributed. 
%and that the DMC is symmetric.
%A channel is said to be half symmetric if the rows of the channel transition matrix, $\left\{P_{Y|X}(y|x)\right\}$, are permutations of each other. 
%, $Y$ is also uniformly distributed.
Under these conditions, Eq. (\ref{funcG}) becomes:
\begin{equation}
\Gamma_y(\vec{p}_{z})={K}^{\alpha-1}\frac{\vec{p}_z \cdot \vec{p}^{\alpha}_{y}}{(\vec{p}_z \cdot \vec{p}_{y})^{\alpha-1}},
\label{G_func}
\end{equation}
where we have defined the following ${K}$-dimensional vectors:
\begin{equation}
\begin{array}{lll}
\vec{p}_{y}&=&[p(y|x),x\in {\cal{X}}],\\
&&\\
\vec{p}^{\alpha}_{y}&=&[p(y|x)^{\alpha},x\in {\cal{X}}].\\
\end{array}
\end{equation}
Applying (\ref{RDF}) and (\ref{G_func}) to (\ref{DPT2}), we get:
\begin{eqnarray}
\nonumber R^Q(D)&=&\displaystyle {{K}}^{\alpha-1}\left[(1-d)^{\alpha}+
\displaystyle \frac{d^{\alpha}} {\left({K}-1\right)^{\alpha-1}}\right]\\
\nonumber &\leq& {K}^{\alpha-2} \sup\left\{ \displaystyle\sum_{y,z}\displaystyle\frac{\vec{p}_z \cdot \vec{p}^{\alpha}_{y}}{(\vec{p}_z \cdot \vec{p}_{y})^{\alpha-1}}\right\}\\
&=&C^Q,
\label{RDB2}
\end{eqnarray}
where the supremum is taken over all sets of positive vectors $\left\{ \vec{p}_z\right\}_{z=1}^M$ that satisfy (\ref{cons1}), in order to carry out continuous optimization. It is easy to see that the optimization is done over a convex set. The functions $\Gamma_y(\vec{p}_{z})$ are neither convex nor concave in this case, but we can use the general form of $I^Q(X;Y,Z)$ given in (\ref{importantI}). By simple substitution under the above conditions, $G_{y}(\vec{p}_z)$ has the following form:
\begin{equation}
G_{y}(\vec{p}_z)={K}^{\alpha-1}\frac{\vec{p}_z^{\alpha} \cdot \vec{p}^{\alpha}_{y}}{(\vec{p}_z \cdot \vec{p}_{y})^{\alpha-1}},
\eeee
where $\vec{p}_z^{\alpha}$ is the vector obtained from $\vec{p}_z$ by raising each element to the power of $\alpha$. Clearly, $G_{y}(\vec{p}_z)=\Gamma_y(\vec{p}_{z})$ for any deterministic encoder. The functions $\{G_{y}(\vec{p}_z)\}$ are convex as shown in Lemma \ref{convexityGy}. Therefore, the supremum of $I^Q(X;Y,Z)$ is attained on the extreme points of the convex set. Finding the supremum requires searching over all such points, i.e, over all sets of binary vectors $\left\{ \vec{p}_z\right\}_{z=1}^M$ that satisfy (\ref{cons1}). The meaning is, of course, returning to discrete optimization and performing it over all deterministic encoders. Seemingly, this makes the mechanism above useless. However, at least for some cases, $C^Q$ can be calculated directly, as shown in the following examples. In addition, we can upper bound $C^Q$ by using (\ref{holder1b}). This upper bound may give us non-trivial bounds, as shown in Example 2. It is also shown to be very useful when handling large alphabets, as demonstrated in Section 3. Finally, notice that optimizing over $\alpha$, separately for each rate, will produce the best lower bound on the achievable distortion at this rate.
\subsubsection*{Example 1}
The symmetric DMC is defined by:
\begin{equation}
p(y|x)=\left\{       \begin{array}{cl} \mu & y=x \\ \epsilon & \text{else}  \end{array} \right.
\label{symchan}
\end{equation}
where $\mu,\epsilon \in [0,1]$, $\mu>\epsilon$, and $\mu+({K}-1)\epsilon=1$.
The distortion measure we use is the Hamming distortion, defined in (\ref{HammDis}). In these conditions, the minimal achievable distortion of a scalar source code with a fixed-rate $R=\log M$, is:
\begin{equation}
D(M)=\epsilon ({K}-M).
\label{th1mindis}
\end{equation}
The proof is given in Appendix E. Knowing the best achievable distortion in this case, we can compare it to the bounds we get from (\ref{RDB2}) to examine their quality.
The generalized capacity (\ref{infsup}) for this channel is given by:
\begin{eqnarray}
\nonumber C^Q &=&{K}^{\alpha-2}\cdot \sup\left\{ \displaystyle\sum_{y,z}\displaystyle\frac{\vec{p}_z \cdot \vec{p}^{\alpha}_{y}}{(\vec{p}_z \cdot \vec{p}_{y})^{\alpha-1}}\right\}\\
\nonumber &=&{K}^{\alpha-2}\epsilon\cdot\sup\left\{\displaystyle\sum_z M_z\cdot \displaystyle\frac{\left(M_z+\mu^{\alpha}/\epsilon^{\alpha}-1 \right)}{\left(M_z+\mu/\epsilon-1\right)^{\alpha-1} }\right.\\
\nonumber &&+\Bigg.({K}-M_z)M_z^{2-\alpha}\Bigg\}\\
&=& {K}^{\alpha-2}\epsilon\cdot \sup\left\{\displaystyle\sum_z q_{\alpha}(M_z)\right\},
\label{q_a_def}
\end{eqnarray}
where $M_z$,  $M_z \in \{1,\ldots, {K}-M+1\}$, is the cardinality of $A_z$, i.e., the number of source symbols that are encoded to $z$. Obviously, $\sum_z M_z= {K}$.
Notice that the supremum is taken over all deterministic encoders, where each encoder is represented by a specific set $\{\vec{p}_z\}_{z=1}^M$ as defined in (\ref{pzdef}). The second equality is proved in Appendix F. The function $q_{\alpha}(M_z)$ is concave for $1<\alpha \leq 2$, and may be concave also out of this range, with dependence on the channel parameters, as shown in Appendix F. When $q_{\alpha}(M_z)$ is concave, we can bound the supremum by taking equal $M_z$'s, i.e., $M_z={K}/M$, $\forall z$, and we get:
\begin{eqnarray}
\nonumber C^Q \leq &
 {K}^{\alpha-1} \epsilon\Bigg[\displaystyle  \displaystyle\frac{\left({K}/M+\mu^{\alpha}/\epsilon^{\alpha}-1 \right)}{\left({K}/M+\mu/\epsilon-1\right)^{\alpha-1} }\Bigg.\\
\nonumber &\\
&\Bigg.+(M-1)\left(\displaystyle\frac{{K}}{M}\right)^{2-\alpha}\Bigg].
\label{lblastchan1}
\end{eqnarray}
If $M$ divides ${K}$, this bound is achieved by any feasible encoder that partitions the source alphabet into equally sized subsets, thus the optimization is exact. The lower bounds on the distortion are obtained by combining (\ref{RDF}) and (\ref{lblastchan1}).
An example for specific values of $\mu$ and $\epsilon$ is presented in Fig. \ref{figb1}. The bound is compared with the bound obtained from the classical DPT (\ref{wzlbtheo22}), the bound obtained from the classical inequality $R_{WZ}(D)\leq \log M$, the bound obtained by using (\ref{holder1b}) and the exact solution of Eq. (\ref{th1mindis}). The WZ RD function was calculated using the Blahut-Arimoto-type algorithm presented in \cite{BAalg}.
Eq. (\ref{RDB2}) was optimized over $\alpha$, for each $M\leq {K}$, so as to get the best lower bound on the distortion.
\begin{figure}[hc]
	\centering		\includegraphics[width=3.5in, height=2in]{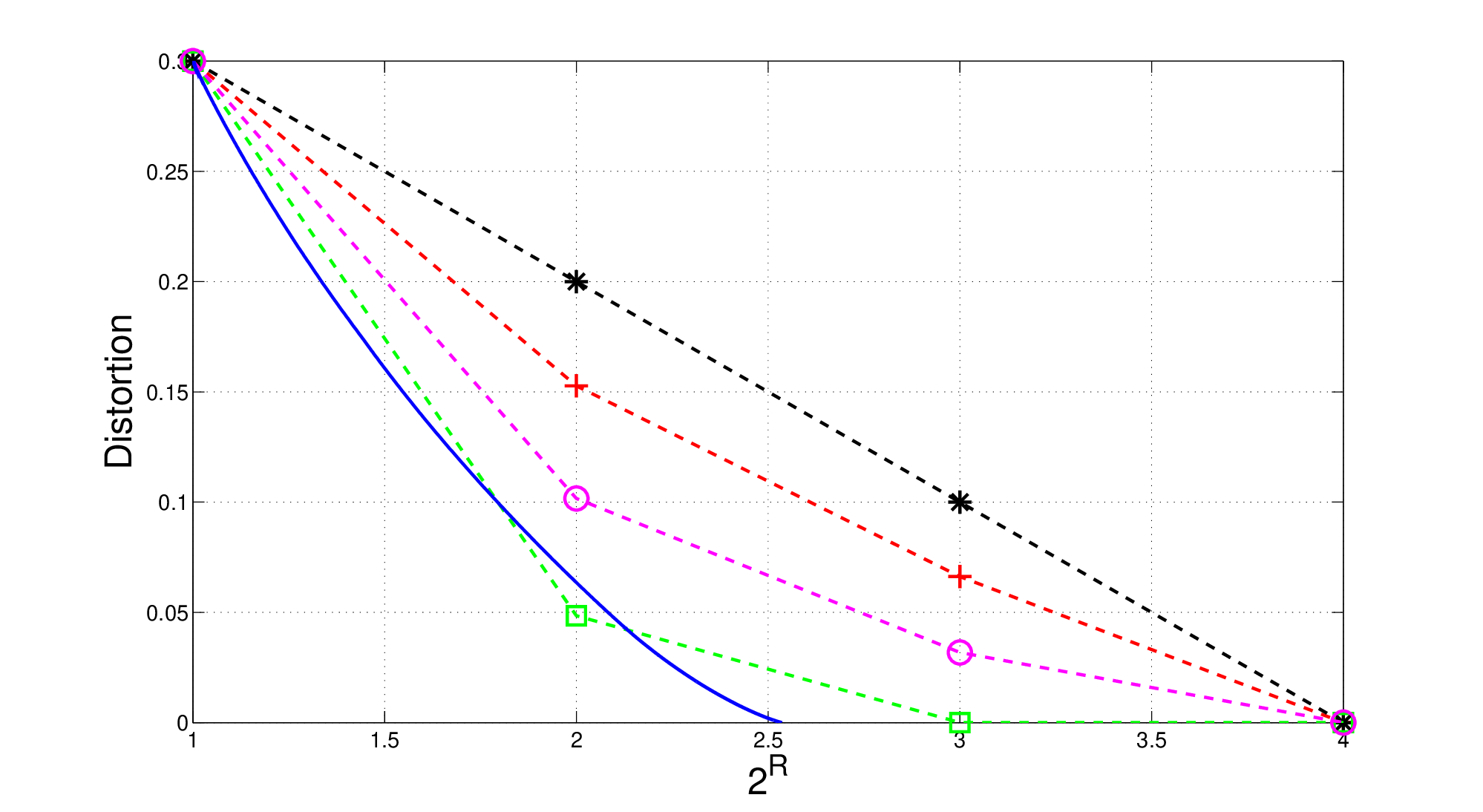} 
	\caption {${K}=4$, $\mu=0.7$. Plus - the lower bound obtained from (\ref{lblastchan1}). Circle - the lower bound obtained from the classical DPT (\ref{wzlbtheo22}). Star - the exact solution. Solid line - the lower bound obtained from $R_{WZ}(D)$. Square - the lower bound obtained from (\ref{holder1b}).}
\label{figb1}
\end{figure}
We see that even the classical DPT gives us non-trivial lower bounds and that the lower bound obtained from (\ref{lblastchan1}) is much better than the trivial bound obtained from $R_{WZ}(D)$. The lower bound obtained from (\ref{holder1b}) is not useful in this case.
There is a gap between the exact solution and the best bound, even for $M=2$, where the optimization (\ref{infsup}) is exact.
\subsubsection*{Example 2}
The symmetric DMC is defined by:
\begin{equation}
p(y|x)=\left\{       \begin{array}{cl} 1/l & y \in \{x\bmod {K},\ldots,(x+l-1)\bmod {K}\}\\ 0 & \text{else}  \end{array} \right.
%\{x,1+x\bmod{{K}},\ldots,1+(x+l-2)\bmod{{K}} \}\\ 0 & else  \end{array} \right.
\label{symchan2}
\end{equation}
where $l$ is an integer, $0<l<{{K}}$, and ${K}\bmod {K}$ is defined to be ${K}$.
Given an input $x$, the channel produces one of $l$ values with equal probability.
The generalized capacity (\ref{infsup}) for this channel is given by:
\begin{eqnarray}
\nonumber C^Q &=& {K}^{\alpha-2}\cdot \sup\left\{ \displaystyle\sum_{y,z}\displaystyle\frac{\vec{p}_z \cdot \vec{p}^{\alpha}_{y}}{(\vec{p}_z \cdot \vec{p}_{y})^{\alpha-1}}\right\}\\
\nonumber &=& {K}^{\alpha-2} \cdot\sup\left\{ \displaystyle\sum_{y,z}\displaystyle\frac{M_{y,z}(1/l)^{\alpha}}{\left(M_{y,z}(1/l)\right)^{\alpha-1}}\right\}\\
&=& {K}^{\alpha-2}\cdot l^{-1}\cdot\sup\left\{ \displaystyle\sum_{y,z}M_{y,z}^{2-\alpha}\right\},
\label{capchan1}
\end{eqnarray}
where $M_{y,z}=l\cdot[\vec{p}_z\cdot \vec{p}_{y}]$. 
%Simply, $M_{y,z}$ is the number of non-zero components in $\vec{p}_{y}$ that overlap to non-zero components of $\vec{p}_z$. 
It is easy to see that $\sum_{z} M_{y,z}= l$. For $1 < \alpha < 2$, the function $M_{y,z}^{2-\alpha}$ is concave in $M_z$. Thus, the supremum is achieved by setting $M_{y,z}=l/M$, $\forall \{y,z\}$:
\begin{equation}
C^Q ={K}^{\alpha-2}\cdot l^{-1}\cdot {K} \cdot M\cdot (l/M)^{2-\alpha}={K}^{\alpha-1}\cdot (M/l)^{\alpha-1}.
\label{CQEXAMPLE2}
\end{equation}
If $M$ divides $l$, equal $M_{y,z}$'s can be obtained by the following feasible encoder:
\begin{equation}
z=f(x)=1+x\bmod{M}.
\label{encdef}
\end{equation}
Therefore, in this case, the optimization is exact. For $\alpha>2$, $C^Q$ is infinite, because we can always set some $M_{y,z}$ to $0$ by an appropriate choice of the encoder. Thus, this range of $\alpha$ does not lead to a useful bound. The lower bounds on the distortion are obtained by combining (\ref{RDF}) and (\ref{CQEXAMPLE2}). An example for specific values of $K$ and $l$ is presented in Fig. \ref{figb2}. The lower bound on the distortion, which coincides with the bound obtained from (\ref{holder1b}) (the upper bound on $C^Q$ is tight for this channel), is compared with the bound obtained from the classical DPT (\ref{wzlbtheo22}) and the bound obtained by the classical inequality $R_{WZ}(D)\leq \log M$. Eq. (\ref{RDB2}) was optimized over $\alpha$, for each $M\leq {K}$, so as to get the best lower bound on the distortion. We see that in this case, the generalized DPT leads to bounds that are better than the trivial bound, whereas the classical DPT does not lead to a useful bound. We also present the exact distortion of the encoder defined in (\ref{encdef}), which is, of course, an upper bound on the distortion. Thus, the distortion of the optimal encoder must be in the range between this upper bound and our highest lower bound.
%For $M=l$, zero distortion can indeed be achieved using the encoder defined in (\ref{encdef}), thus our lower bound at this point %is tight.
\begin{figure}[hc]
	\centering		
	\includegraphics[width=3.5in, height=2in]{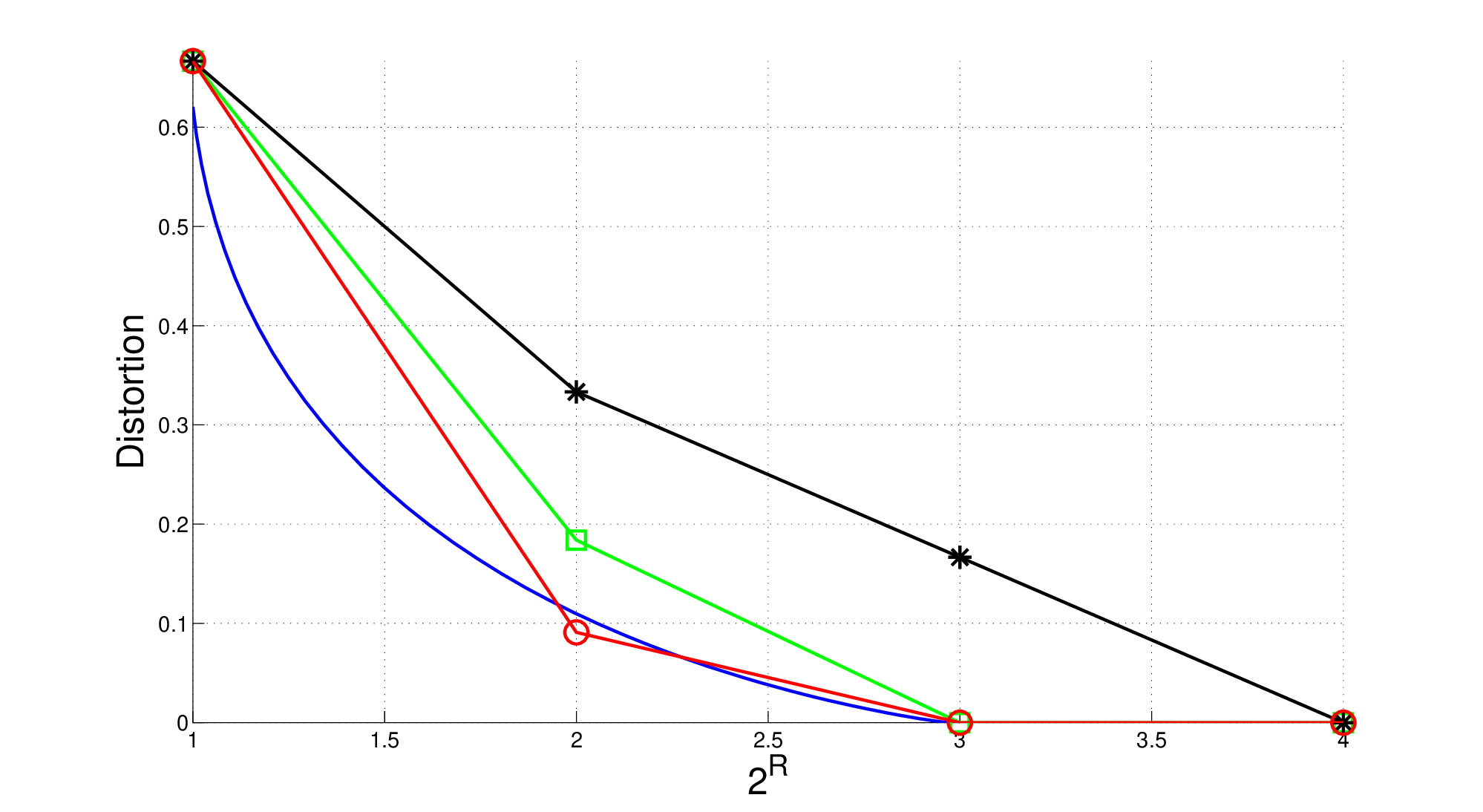} 
	\caption {${K}=4$, $l=3$. Square - the lower bound obtained from (\ref{lblastchan1}) and (\ref{holder1b}). Circle - the lower bound obtained from the classical DPT (\ref{wzlbtheo22}). Solid line - the lower bound obtained from $R_{WZ}(D)$. Star - the exact distortion of the encoder (\ref{encdef}).}
\label{figb2}
\end{figure}%%%%%%%%%%%%%%%%%%%%%%%%%%%%%%%%%%%%%%%%%%%%%%%%%%%%%%%%%%%
\subsubsection*{Large alphabets}
In this part we show that as the alphabet size increases, we obtain interesting bounds on the performance of scalar coding. These useful bounds can be obtained for a large variety of channels, without any symmetry requirements. The results are obtained using the upper bound on the ``capacity'', presented in Lemma \ref{holder1btheorem}. This bound becomes tighter as the alphabet size increases, for various channels. For these channels, we can get close to the last upper bound in Eq. (\ref{holder_bound}), i.e., to achieve $\ddss\{\displaystyle \vec{p}_z\cdot [p(x_1)\cdot p(y|x_1)^{\frac{1}{2-\alpha}}, p(x_2)\cdot p(y|x_2)^{\frac{1}{2-\alpha}}, \ldots]\}_{z=1}^M$ which are almost equal to each other, by a suitable choice of encoder. This is because $\{p(y|x)\}$ is composed of large number of probabilities with small values.
Using the bound of Lemma \ref{holder1btheorem}, we bypass the problem of optimizing the capacity for general channels. As was mentioned earlier in Subsection 2.3, this optimization is in general a convex maximization problem which requires searching over all possible encoders. Using our lower bounds along with simple upper bounds, we give the performance range for scalar coding. These results are, of course, interesting from the practical point of view. 

In the following examples we assume that the sources are uniform. This is because analytic expression for the generalized RD function of general sources is not available. We use the Hamming distortion measure for convenience. Bounds for more general distortion measures can be calculated using the result of Lemma \ref{RDgenerallemma}. In the two former examples, we compare our results to the WZ RD function. This function was calculated using the algorithm presented in \cite{BAalg}. However, the computational complexity of this algorithm is of order ${K}^{K}$. Therefore, this algorithm is not practical for large alphabets. Since no other efficient algorithms are known, the computation of the WZ RD function for large alphabets is problematic. As a result, even the trivial bound obtained from $R_{WZ}(D)\leq \log M$ does not lead to a closed-form expression. This makes our results even more interesting. 

Instead of comparing our distortion bounds to the bounds obtained from $R_{WZ}(D)\leq \log M$, we compare them to the following linear function:
\bbee
\overline{D}(R)=D_{max} \left(1-\ddss\frac{R}{H(X|Y)}\right),
\eeee
where $D_{max}=R^{-1}_{WZ}(0)$, i.e., the lowest achievable distortion for rate $R=0$. 
The function $\overline{D}(R)$ is simply the straight line obtained by time-sharing the two known endpoints of the $R^{-1}_{WZ}(R)$ curve, $(0,D_{max})$ and $(H(X|Y),0)$.  This line is, of course, a trivial upper bound on $R^{-1}_{WZ}(R)$. As an upper bound on the best achievable distortion of a fixed-rate code, we use the performance of the code composed of the encoder defined in (\ref{encdef}), along with the corresponding optimal decoder, given by:
\begin{equation}
\hat{x}=g(z,y)=\displaystyle\argmax_{x\in A_z} \{ p(y|x) \}.
\end{equation} 
Remember that the optimal decoding strategy is maximum likelihood because we use the Hamming distortion measure. The choice of the encoder (\ref{encdef}) seems natural when handling channels with transition probabilities that decrease with the distance. In this case, we want adjacent symbols to be in different subsets of the encoder. Obviously, any code that performs better will improve the performance range. 

In the first example, the DMC is defined by:
\begin{equation}
p(y|x)=c_x\exp\left\{-\ddss\frac{(x-y)^2}{2K^2\sigma^2}\right\},
\label{chanlargealphabet1}
\end{equation}
where $c_x$ is a normalization factor such that $\sum_y p(y|x)=1$. This is a 'Gaussian'-like channel. Notice that this channel is not symmetric. Performance ranges for different values of $K$ are presented in Fig. 4. We see that we get bounds that are higher than $\overline{D}(R)$ and therefore, higher than $R^{-1}_{WZ}(R)$. These bounds also show that the time-sharing of $\overline{D}(R)$ performs better than any fixed-rate code for considerable range of rates. Similar results can be presented for various channels, not necessarily additive.

In the second example, the DMC is the same as in Example 2 and is defined by (\ref{symchan2}). In this case, $H(X|Y)=\log l$. We now present the performance range for large alphabets where in all cases we take $l=K/4$. The results are shown in Fig. 5. We see that our bounds are higher than $\overline{D}(R)$ for $R=\log M>\log3$. As the alphabet size increases, the gap between our bound and $\overline{D}(R)$ also increases.

In all examples, one can easily notice that the lower convex envelope of our lower bounds, is very close to the straight line $\overline{D}(R)$. As was shown in Section 2, this convex envelope is a lower bound on the distortion of time-sharing fixed-rate codes.

\begin{figure}[ht]
 \centering
 \subfigure[$K=64$, $\sigma=0.5$]{
  \includegraphics[width=3.5in, height=2.5in]{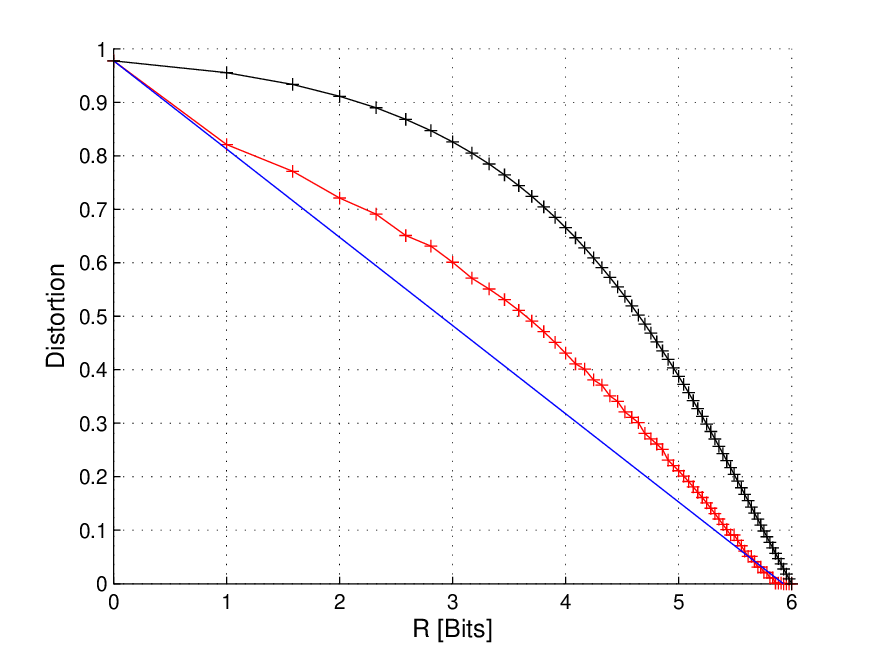}
   \label{fig:subfig1}
   }
 \subfigure[$K=128$, $\sigma=0.5$]{
  \includegraphics[width=3.5in, height=2.5in]{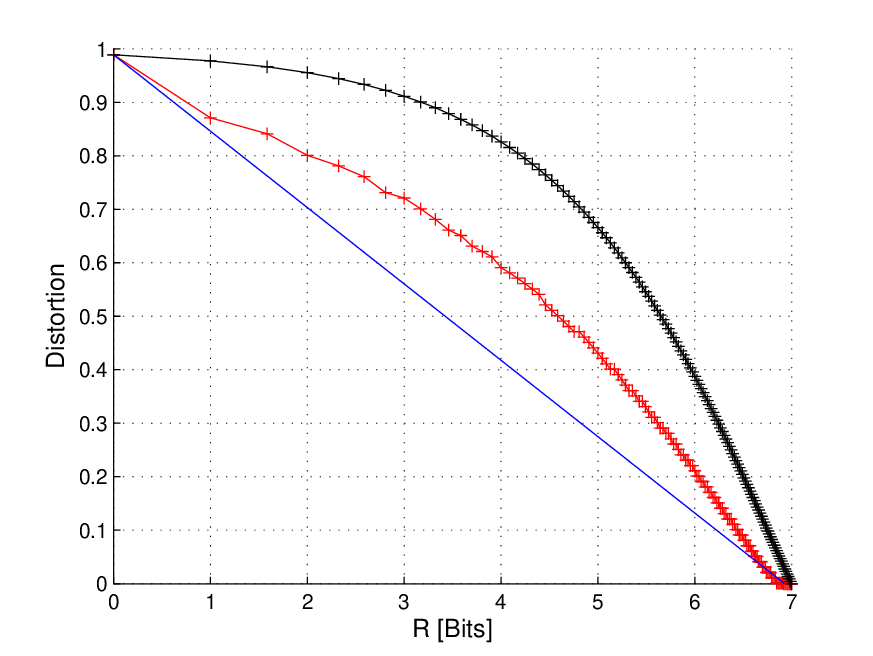}
   \label{fig:subfig2}
   }
 \subfigure[$K=256$, $\sigma=0.5$]{
  \includegraphics[width=3.5in, height=2.5in]{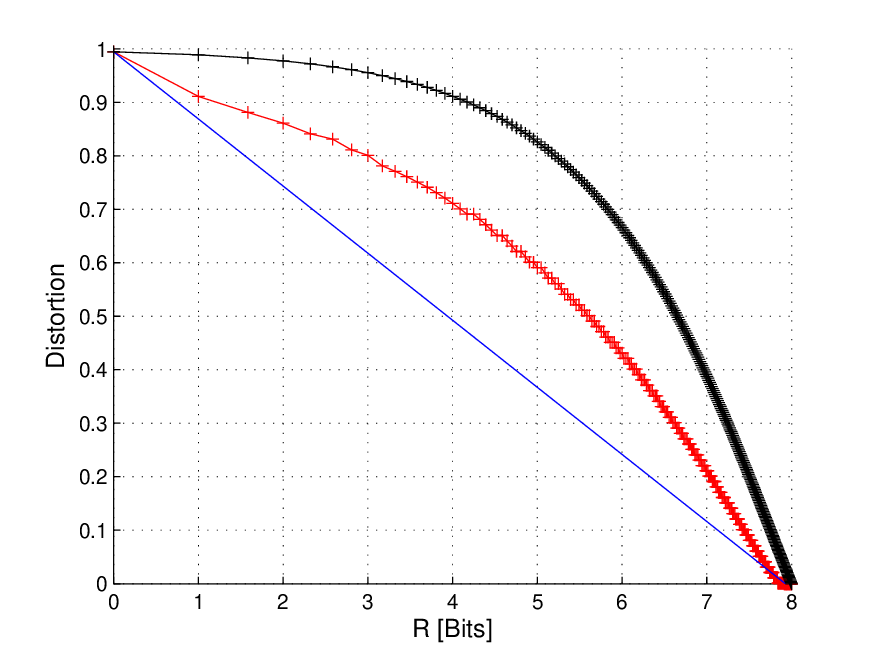}
   \label{fig:subfig3}
   }
 \label{largealphabet}
 \caption[]{The performance range for the channel (\ref{chanlargealphabet1}) and different sizes of alphabets. Straight line - $\overline{D}(R)$. Lower curve - our lower bound. Higher curve - the distortion of the code (\ref{encdef}).}
\end{figure}
\begin{figure}[ht]
 \centering
 \subfigure[$K=64$, $\l=16$]{
  \includegraphics[width=3.5in, height=2.5in]{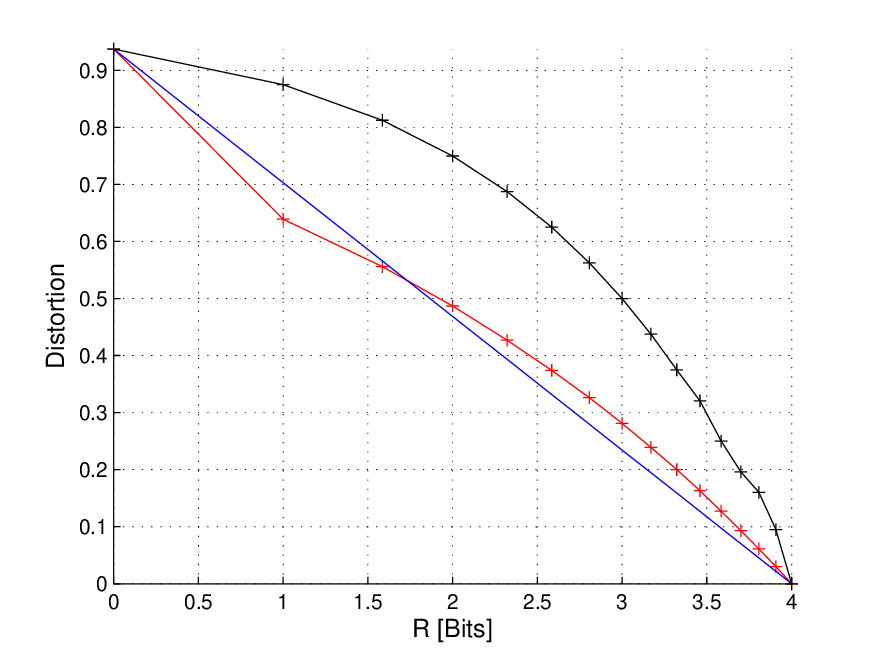}
   \label{fig:subfig1}
   }
 \subfigure[$K=128$, $\l=32$]{
  \includegraphics[width=3.5in, height=2.5in]{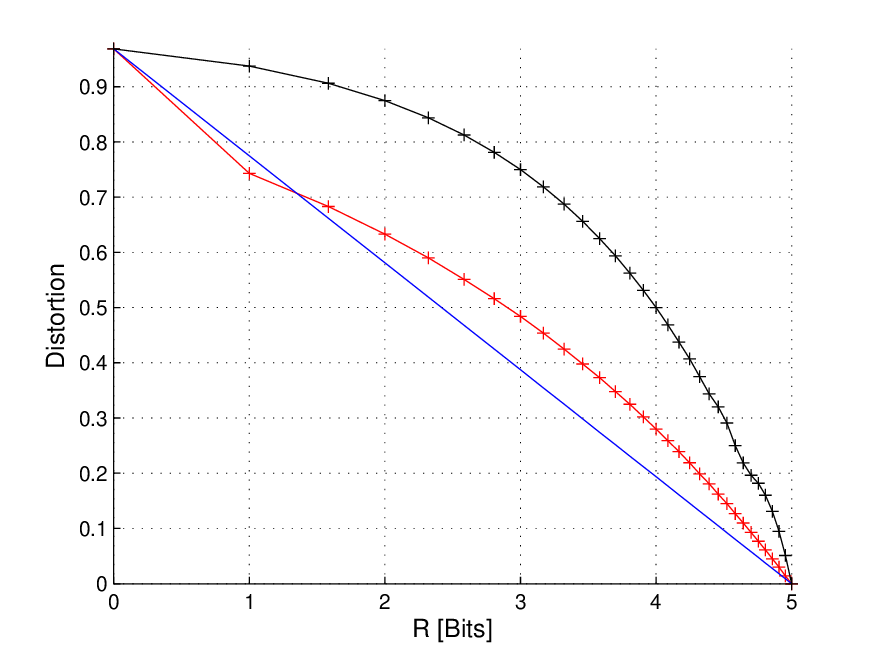}
   \label{fig:subfig2}
   }
 \subfigure[$K=256$, $\l=64$]{
  \includegraphics[width=3.5in, height=2.5in]{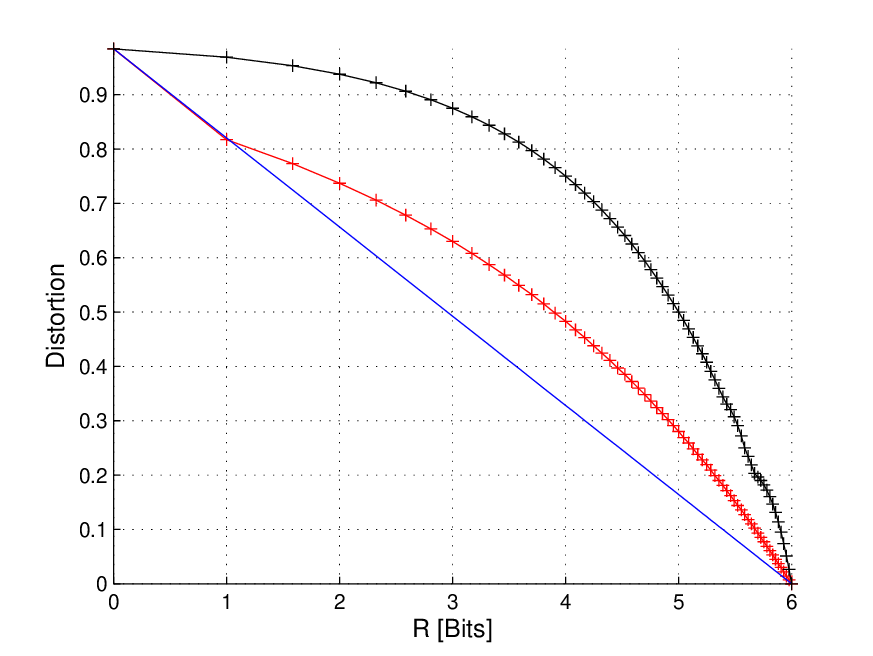}
   \label{fig:subfig3}
   }
 \label{largealphabet}
 \caption[]{The performance range for the channel (\ref{symchan2}) and different sizes of alphabets and $l$'s. Straight line - $\overline{D}(R)$. Lower curve - our lower bound. Higher curve - the distortion of the code (\ref{encdef}).}
\end{figure}
\clearpage
\section{Future Directions}
%In this paper, we presented the relevant random variables of the WZ problem as a Markov chain. As far as we know, this is a new formulation in the context of the WZ setting, which enabled us to use the DPT and its generalizations. We then focused on the generalized DPT presented in \cite{ZZ1}. We found the generalized RD function of uniform sources, for any convex function $Q$ and a broad class of distortion measures. We also found a useful upper bound for the generalized capacity in the setting above and calculated this capacity exactly for several interesting channels. We then showed that replacing the logarithmic function with other functions, in the DPT we formulated, yield better bounds on the distortion of scalar coding in the WZ setting. Examples of non-trivial lower bounds for scalar coding in this setting were given. For large alphabets, we demonstrated that non-trivial bounds can be achieved for various channels and as a result, the performance range for scalar coding can be given. We also saw that simple time-sharing between the two endpoints of the WZ RD curve, may perform better than any fixed-rate code in various cases.
%As far as we know, these bounds are the only existing non-trivial bounds for this situation. Clearly, these results are relevant from the practical point of view. 
Clearly, the results of these paper are relevant from the practical point of view and so are their possible extensions. Analytic expressions for the generalized RD function of general sources are not apparent to be available. Improving the bounds above by using, for example, the techniques of \cite{Merhav1}, is yet to be explored. In the next step, a possible direction will be to extend our setting to more general scenarios. The first interesting scenario is the variable-rate coding case, where $Z$ is encoded by a variable-length code. In this setting, 
%$Z$ will be allowed to take on values in $\{1,\ldots,K\}$. 
the generalized capacity (\ref{infsup}) will be maximized under the constraint $\mathbb{E}\{L(Z)\}\leq R$, where $L(Z)$ is the length of the codeword $Z$. The best variable-rate code that can be used is the Huffman code for $p(z)$, provided that $p(x,y)>0$ for all $x,y\in{\cal{X}}$. The challenge is to perform this optimization.
Another interesting scenario is when the output of the encoder is transmitted over a noisy channel (instead of the noiseless one in the WZ setting). In this case, the generalized capacity will be of the form $I^Q(X,Z;Z',Y)$, where $Z'$ is the output of the encoder noisy channel. Again, the challenge will be to optimize this capacity. %Another direction is to extend the mechanism above to settings of coding with memory. In these scenarios, the states of the encoder and the decoder will be allowed to depend on past inputs. Using generalized DPTs, we hope to get useful bounds for this case. 
\newpage
\section*{Appendix A - Proof of Lemma \ref{convexityGy}}
\renewcommand{\theequation}{A.\arabic{equation}}
\setcounter{equation}{0}
\begin{proof}
We defined the following functions $\left\{G_{y}(\vec{p}_z)\right\}_{y=1}^K$, $G_y: {\mathbb{R}}_+^{K} \rightarrow \mathbb{R}$:
\begin{equation}
G_{y}(\vec{p}_z)=\ddss\sum_{{x}}p({x})p({y}|{x})p(z|{x})Q\left(\ddss\frac{\ddss\vec{p}_z\cdot \vec{p}_{{y}}}{p({y}|{x})p(z|{x})}\right).
\end{equation}
In order to prove the convexity of $G_y(\vec{p}_z)$, we use the following known result (cf., e.g., \cite{ConvexOptimization}):
$G_y(\vec{p}_z)$ is convex if and only if $\mbox{dom}(G_y)$ is convex, and the function $g:\mathbb{R}\rightarrow\mathbb{R}$, defined by:
\begin{eqnarray}
g(t)=G_y(\vec{r}+\vec{v}t)&,& \mbox{dom}(g)=\left\{t:\vec{r}+\vec{v}t\in \mbox{dom}(G_y)\right\},
\end{eqnarray}
is convex in $t$, for any $\vec{r}\in \mbox{dom}(G_y)$, $\vec{v}\in {\mathbb{R}}^{K}$.
Substituting $\vec{p}_z=\vec{r}+\vec{v}t$, we get:
\begin{eqnarray}
\nonumber g(t)&=&\ddss\sum_{x=1}^{K}p(x)p(y|x)(r_x+v_x t)Q\left(\ddss\frac{\ddss[\vec{r}\cdot\vec{p}_{y}]+[\vec{v}\cdot \vec{p}_{y}]t} {p(y|x)(r_x+v_x t)}\right)\\
&=&\ddss\sum_{x=1}^{K}p(x)p(y|x)(r_x+v_x t)Q\left(\ddss\frac{\ddss a_{y}+b_{y}t} {p(y|x)(r_x+v_x t)}\right),
\end{eqnarray}
where $\{r_x\}_{x=1}^{K}$ and $\{v_x\}_{x=1}^{K}$ are the elements of $\vec{r}$ and $\vec{v}$, respectively, and $a_y\equiv\vec{r}\cdot \vec{p}_{y}$, $b_y\equiv\vec{v}\cdot \vec{p}_{y}$ are constants. Linear combination of convex functions is convex.
Thus, it is enough to show the convexity of the following functions, $x=\left\{1,\ldots, K \right\}$:
\begin{eqnarray}
\nonumber f_x(t)&=&\ddss p(y|x)(r_x+v_x t)Q\left(\ddss\frac{\ddss a_{y}+b_{y}t} {p(y|x)(r_x+v_x t)}\right)\\
&=&h\left(a_y+b_y t,p(y|x)(r_x+v_x t)\right),
\end{eqnarray}
where $h(u,s)=\ddss s Q\left(u/s\right)$ is the perspective (cf., e.g., \cite{ConvexOptimization}) of the convex function $Q(u)$, and thus convex. The function $f_x(t)$ is the restriction of the convex function $h(u,s)$ to the straight line\\ 
$\left\{u=p(y|x)(r_x+v_x t), s=a_y+b_y t\right\}$. Therefore, $f_x(t)$ is convex.
\end{proof}
%%%%%%%%%%%%%%%%%%%%%%%%%%%%%%%%%%%%%%%%%%%%%%%%%%%%%%%%%%%%%%%%%%%%%%%%%%%%%%%%%%%%%%%%%%
\section*{Appendix B - Proof of Eq. (\ref{boundclassicholder})}
\renewcommand{\theequation}{B.\arabic{equation}}
\setcounter{equation}{0}
\begin{proof}
It is enough to calculate the limit:
\begin{eqnarray}
\nonumber\displaystyle\lim_{\alpha\rightarrow 1}\displaystyle\frac{1}{\alpha-1}\log\left(\displaystyle\sum_y    \displaystyle  \left(\displaystyle \sum_{x} p(x)\cdot p(y|x)^{\frac{1}{2-\alpha}} \right)^{2-\alpha}\right)\\
=\displaystyle\lim_{\alpha\rightarrow 1}\displaystyle \ddss\frac{d}{d\alpha}\left[\log\left(\displaystyle\sum_y    \displaystyle  \left(\displaystyle \sum_{x} p(x)\cdot p(y|x)^{\frac{1}{2-\alpha}} \right)^{2-\alpha}\right)\right],
\end{eqnarray}
where the equality is due to L'H\^{o}pital's rule, using the fact that both the numerator and denominator tend to $0$ as $\alpha\rightarrow 1$. Noticing that the expression in the brackets has the form of the Gallager function (see, e.g., \cite[Chap. 5]{Gallager}) for the given DMC, defined as
\bbee
E_0(\rho,\mathbf{p})=-\log\left(\displaystyle\sum_y    \displaystyle  \left(\displaystyle \sum_{x} p(x)\cdot p(y|x)^{\frac{1}{1+\rho}} \right)^{1+\rho}\right),
\eeee
we can use the following property of $E_0(\rho,\mathbf{p})$:
\bbee
\left.\frac{\partial E_0(\rho,\mathbf{p})}{\partial \rho} \right|_{\rho=0}= I(X;Y)
\eeee
to get:
\bbee
\displaystyle \ddss\frac{d}{d\alpha}\left.\left[\log\left(\displaystyle\sum_y    \displaystyle  \left(\displaystyle \sum_{x} p(x)\cdot p(y|x)^{\frac{1}{2-\alpha}} \right)^{2-\alpha}\right)\right]\right|_{\alpha=1}=I(X;Y),
\eeee
which completes the proof.
\end{proof}
%%%%%%%%%%%%%%%%%%%%%%%%%%%%%
\section*{Appendix C - Proof of Lemma \ref{RDgenerallemma}}
\renewcommand{\theequation}{C.\arabic{equation}}
\setcounter{equation}{0}
The idea of the proof is to exploit the symmetry of the distortion matrix. From symmetry, we expect that each input symbol will have the same set of transition probabilities.
\begin{proof}
We define $p(i|j)\equiv P_{\hat{X}|X}(i|j)$. By definition:
\begin{eqnarray}
\nonumber I^Q(X;\hat{X})&=&\displaystyle\sum_{x,\hat{x}}p(x,\hat{x})Q\left(\frac{p(\hat{x})}{p(\hat{x}|x)}\right)\\
&=&\displaystyle\frac{1}{K}\sum_{i=1}^{K}p(i|i) Q\left(\frac{p(\hat{x}=i)}{p(i|i)}\right)+\displaystyle\frac{1}{K}\ddss\sum_{i=1}^{K}\sum_{j\neq i}^{K}p(j|i) Q\left(\frac{p(\hat{x}=j)}{p(j|i)}\right).
\label{proofrd1}
\end{eqnarray}
Each row of the distortion matrix contains the same $K$ values. We enumerate these values as $\{\rho_1,\rho_2,\ldots,\rho_{K}\}$ where $\rho_1\equiv \rho(x,x)=0$. Without loss of generality and only for convenience of the proof, we assume that we have $K$ different values.
We define:
\bbee
\tilde{x}_k(x)=\{\tilde{x}\in{{\cal{X}}}:\rho(x,\tilde{x})=\rho_k\}.
\eeee
In words, $\tilde{x}_k(x)$ is the unique alphabet symbol with distortion $\rho_k$ relative to $x$. Using this definition, we can write ($\ref{proofrd1}$) in the following way:
\begin{eqnarray}
\nonumber I^Q(X;\hat{X})&=&\displaystyle\frac{1}{K}\sum_{i=1}^{K}p(i|i) Q\left(\frac{p(\hat{x}=i)}{p(i|i)}\right)+\displaystyle\frac{1}{K}\sum_{i=1}^{K}\sum_{k=2}^{K}p(\tilde{x}_k(i)|i) Q\left(\frac{p(\hat{x}=\tilde{x}_k(i))}{p(\tilde{x}_k(i)|i)}\right)\\
\label{RDGEN11111}
\end{eqnarray}
We define:
\bbee
\bbaa{lllll}
p_k&=&\displaystyle\frac{1}{K}\ddss\sum_{i=1}^{K}p(\tilde{x}_k(i)|i),&k\in\{2,\ldots,K\}\\
p_1&=&\displaystyle\frac{1}{K}\ddss\sum_{i=1}^{K}p(i|i)=1-\ddss\sum_{k=2}^{K}p_k &&
\eeaa
\eeee
Using these definitions, the distortion is given by:
\bbee
D=\displaystyle\frac{1}{K}\ddss\sum_{k=2}^{K}\ddss\sum_{i=1}^{K}p(\tilde{x}_k(i)|i)\rho_k=\ddss\sum_{k=2}^{K}p_k \rho_k\\
\label{consRDproof}
\eeee
Applying the Jensen inequality, we get:
\begin{eqnarray}
\nonumber I^Q(X;\hat{X})&=&\displaystyle\frac{1}{K}\sum_{i=1}^{K}p(i|i) Q\left(\frac{p(\hat{x}=i)}{p(i|i)}\right)+
\displaystyle\frac{1}{K}\sum_{k=2}^{K}\sum_{i=1}^{K}p(\tilde{x}_k(i)|i)
 Q\left(\frac{p(\hat{x}=\tilde{x}_k(i))}{p(\tilde{x}_k(i)|i)}\right)\\
\nonumber &=&p_1 \ddss\sum_{i=1}^{K}\ddss\frac{p(i|i)}{K p_1 } Q\left(\ddss\frac{p(\hat{x}=i)}{p(i|i)}\right)+\ddss\sum_{k=2}^{K}p_k\sum_{i=1}^{K}\frac{p(\tilde{x}_k(i)|i)}{K p_k} Q\left(\frac{p(\hat{x}=\tilde{x}_k(i))}{p(\tilde{x}_k(i)|i)}\right)\\
\nonumber &\geq&\displaystyle p_1  Q\left(\sum_{i=1}^{K}\frac{p(i|i)p(\hat{x}=i)}{K p_1 p(i|i)}\right)+
\ddss\sum_{k=2}^{K}p_kQ\left(\sum_{i=1}^{K}\frac{p(\tilde{x}_k(i)|i)p(\hat{x}=\tilde{x}_k(i))}{K p_k p(\tilde{x}_k(i)|i)}\right)\\
\nonumber &=&p_1 \displaystyle Q\left(\frac{1}{K p_1 }\right)
+\ddss\sum_{k=2}^{K}p_k Q\left(\frac{1}{K p_k }\right)\\
&=&\displaystyle\sum_{k=1}^{K}p_k Q\left(\ddss\frac{1}{K p_k}\right).
\label{RDproof5}
\end{eqnarray}
Notice that the sum $\ddss\sum_{i=1}^{K}p(\hat{x}=\tilde{x}_k(i))$ is running over all values of $\hat{x}$ due to the symmetry of the distortion matrix ($\rho_k$ appears in each column) and thus equal to $1$. The lower bound in (\ref{RDproof5}) is achieved by a channel of the form:
%$\left\{P_{\hat{X}|X}\right\}$ 
\begin{eqnarray}
p(\tilde{x}_k(i)|i)=p_k,&k\in\{2,\ldots,K\},&i\in \{1,2,\ldots,K\}.
\label{chanrdproof}
\end{eqnarray}
This channel achieves, of course, the same distortion $D$, which depends only on $\{p_k\}$.
We also have for this channel:
\begin{eqnarray}
p(\hat{x}=i)=\ddss\frac{1}{K}p(i|i)+\ddss\frac{1}{K}\ddss\sum_{k=2}^{K} p_k=\frac{1}{K}
\label{equRDgen111}
\end{eqnarray}
Substituting (\ref{chanrdproof}) and (\ref{equRDgen111}) in (\ref{RDGEN11111}), we get:
\begin{eqnarray}
\nonumber I^Q(X;\hat{X})&=&\displaystyle\frac{1}{K}\sum_{i=1}^{K}p(i|i) Q\left(\frac{p(\hat{x}=i)}{p(i|i)}\right)+\displaystyle\frac{1}{K}\sum_{i=1}^{K}\sum_{k=2}^{K}p(\tilde{x}_k(i)|i) Q\left(\frac{p(\hat{x}=\tilde{x}_k(i))}{p(\tilde{x}_k(i)|i)}\right)\\
\nonumber &=&\displaystyle\frac{1}{K}\sum_{i=1}^{K}p_1 Q\left(\frac{1}{K p_1}\right)+\displaystyle\frac{1}{K}\sum_{i=1}^{K}\sum_{k=2}^{K}p_k Q\left(\frac{1}{K p_k}\right)\\
&=&\displaystyle\sum_{k=1}^{K}p_k Q\left(\ddss\frac{1}{K p_k}\right),
\label{rdprooflowerbound}
\end{eqnarray}
which is exactly the lower bound. In summary, we showed that the channel that minimizes $I^Q(X;\hat{X})$, among all channels with the same $\{p_k\}_{k=1}^{K}$, is of the form (\ref{chanrdproof}). Therefore, to get the RD function, it is enough to optimize (\ref{rdprooflowerbound}) over all probability measures $\{p_k\}_{k=1}^{K}$, subject to the constraint (\ref{consRDproof}). 
The generalized mutual information $I^Q(X;\hat{X})$ is, of course, convex in $\{p_k\}_{k=1}^{K}$. Thus, this is a standard convex minimization under linear constraints problem and the solution is given by the Karush-Kuhn-Tucker conditions, which are exactly (\ref{KKT1}) and (\ref{KKT2}). The lower bound obtained must be convex in $D$, otherwise we could achieve its lower convex envelop by linear combination of $\{p_k\}$'s. In addition, it has a minimum equal to $Q(1)$ at $D=\ddss\frac{1}{K}\sum_{k=1}^K\rho_k$. This minimum is achieved by taking uniform $\{p_k\}$, which agrees with $\lambda_2=0$ in Eq. (\ref{KKT1}). Since the lower bound is decreasing for $D\leq\ddss\frac{1}{K}\sum_{k=1}^K\rho_k$, the distortion constraint $\mathbb{E} \rho(X,\hat{X}) \leq D$ is indeed achieved with equality in this range. For $D>\ddss\frac{1}{K}\sum_{k=1}^K\rho_k$, we of course have $R^Q(D)=Q(1)$.
\end{proof}
%%%%%%%%%%%%%%%%%%%%%%%%%%%%%%
\section*{Appendix D - Proof of Lemma \ref{RDHamming123}}
\renewcommand{\theequation}{D.\arabic{equation}}
\setcounter{equation}{0}
\begin{proof}
Again, the idea of the proof is to exploit the symmetry of the source and the Hamming distortion.
We assume that the source $X$ is uniformly distributed over ${\cal{X}}$, and use the Hamming distortion measure.
Under these conditions, the distortion $D$ is given by:
\begin{eqnarray}
\nonumber D &=&\ddss\sum_{i,j =1}^{{{K}}} p(i) p(j|i) \rho(i,j)\\
\nonumber &=&\ddss \frac{1}{{K}} \sum_{j}\sum_{i\neq j} p(j|i)\\
&=&\ddss \frac{1}{{K}} \sum_{j} D_j,
\end{eqnarray}
where we have defined $p(i)=P_X(i)$, $p(j|i)=P_{\hat{X}|X}(j|i)$, and:
\begin{equation}
D_j=\sum_{i\neq j} p(j|i).
\end{equation}
Thus we have the following:
\begin{equation}
p(\hat{x}=j)=\frac{1}{{K}}\sum_i p(j|i) =\frac{1}{{K}}\sum_{i\neq j} p(j|i)+\frac{1}{{K}} p(j|j) = \frac{1}{{K}} \left[ D_j+ p(j|j) \right]
\label{rd1}
\end{equation}
and
\begin{equation}
\begin{array}{rllll}
 1-D&=&\ddss\frac{1}{{K}}\sum_i p(i|i)&&\\
 D&=&\ddss\frac{1}{{K}}\sum_j\sum_{i\neq j} p(j|i)&=&\ddss\frac{1}{{K}}\sum_i\sum_{j\neq i} p(j|i).
\end{array}
\label{rd2}
\end{equation}
Calculating :
\begin{eqnarray}
\nonumber I^Q(X;\hat{X})&=&\ddss\sum_{x,\hat{x}}p(x,\hat{x})Q\left(\frac{p(\hat{x})}{p(\hat{x}|x)}\right)\\
\nonumber &=&\ddss\frac{1}{K}\sum_{i} p(i|i)Q\left[\frac{1}{{K}}\ddss \left(\frac{D_i}{p(i|i)}+1\right)\right]+
\ddss\frac{1}{K}\sum_{i}\sum_{j\neq i} p(j|i)Q\left[\frac{\ddss D_j+p(j|j)}{Kp(j|i)}\right]\\
\nonumber &\geq&\ddss (1-D) \cdot Q\left[\frac{1}{(1-D){K}^2}\sum_{i} \left( D_i+p(i|i)\right) \right]+\\
\nonumber &&\ddss D\cdot Q\left[\frac{1}{D{K}^2}\sum_i\sum_{j\neq i}p(j|i)\left(\frac{D_j+p(j|j)}{p(j|i)}\right)\right]\\
\nonumber &=&\ddss (1-D) \cdot Q\left[\frac{D+1-D}{K(1-D)}       \right]+
\ddss D\cdot Q\left[\frac{1}{DK}\sum_i \left( D-\frac{D_i}{{K}}+1-D-\frac{p(i|i)}{{K}} \right)\right]\\
&=&\ddss (1-D) \cdot Q\left[ \frac{1}{K(1-D)}         \right]+
\ddss D\cdot Q\left[   \frac{ {K}-1 }{{KD}}\right].
\end{eqnarray}
The second equality is due to (\ref{rd1}). The inequality is obtained by applying the Jensen inequality to each one of the two weighted sums, after normalizing the weights using (\ref{rd2}). The subsequent equalities are obtained by calculating the sums that appear as arguments of the function $Q$ using (\ref{rd2}) and by simple algebraic manipulations. 
It is easy to show, by simple substitution in $I^Q(X;\hat{X})$, that this lower bound is achieved , for any convex function $Q(t)$, by the following symmetric channel:
\begin{equation}
p(\hat{x}|x)=\left\{\begin{array}{lll} 1-D& \hat{x}=x      \\  \ddss      \frac{D}{{K}-1}  &    \hat{x}\neq x \end{array} \right.
\end{equation}
Following the same steps as in the end of the proof of Lemma \ref{RDgenerallemma}, we get (\ref{RDF}). 
\end{proof}
%%%%%%%%%%%%%%%%%%%%%%%%%%%%%%%%
\section*{Appendix E - Proof of Eq. (\ref{th1mindis})}
\renewcommand{\theequation}{E.\arabic{equation}}
\setcounter{equation}{0}
\begin{proof}
We assume that the source is uniformly distributed and that the DMC is given by (\ref{symchan}). The Hamming distortion is equal to the average probability of error. Therefore, given a scalar encoder of fixed rate $R=\log M$, the optimal decoding strategy is, of course, maximum likelihood:
\begin{equation}
\hat{x}=g(z,y)=\displaystyle\argmax_{x\in A_z} \{ p(y|x) \}
\end{equation}
For the channel (\ref{symchan}), the decoder gets the form:
\begin{equation}
\hat{x}=\left\{ \begin{array}{lll} y& y\in A_z \\ \text{choose } x' \in A_z\ \text{uniformly at random} & y\notin A_z  \end{array} \right.
\end{equation}
Given $x\in A_z$, we have two error events. The first error event is when $y\in A_z$ and $y\neq x$.
The probability of such an event is $(M_z-1)\epsilon$, where $M_z$,  $M_z \in \{1,\ldots, {K}-M+1\}$, is the cardinality of $A_z$, i.e., the number of source symbols that are encoded to $z$. The other error event is when $y\notin A_z$ and $x'\neq x$. The probability of this event is the product $\Pr\{y \notin A_z|x\}\cdot \Pr\{x'\neq x|x\} = ({K}-M_z)\epsilon \cdot (M_z-1)/M_z$. Thus, the distortion is given by:
\begin{eqnarray}
\nonumber d&=&\displaystyle \frac{1}{{K}}\displaystyle\sum_{x} \Pr\{\text{error}|x\}\\ 
\nonumber &=& \displaystyle \frac{1}{{K}}\displaystyle\sum_{z}\sum_{x\in A_z}[(M_z-1)\epsilon+({K}-M_z)\epsilon (M_z-1)/M_z]\\
\nonumber &=& \displaystyle \frac{\epsilon}{{K}}\displaystyle\sum_{z}M_z [(M_z-1)+({K}-M_z) (M_z-1)/M_z]\\
\nonumber &=& \displaystyle \frac{\epsilon}{{K}}\displaystyle\sum_{z}[M_z\cdot(M_z-1)+({K}-M_z) (M_z-1)]\\
\nonumber &=& \displaystyle \frac{\epsilon}{{K}}\displaystyle\sum_{z} (M_z-1) [M_z+{K}-M_z]\\
\nonumber &=& \displaystyle \epsilon \displaystyle\sum_{z} (M_z-1)\\
\nonumber &=& \displaystyle \epsilon \displaystyle (K-M)\\
&=& \displaystyle \epsilon \displaystyle (K-2^R).
\end{eqnarray}
Notice that $d(R)$ is a decreasing concave function. Thus, by time-sharing an encoder with $R=\log K$ and an encoder with $R=0$, we can achieve the straight line $\tilde{d}(R)=\epsilon(K-1)(1-R/\log K)$ and outperform any fixed-rate encoder. 
\end{proof}
%%%%%%%%%%%%%%%%%%%%%%%%%%%%%%%%
\section*{Appendix F - The concavity of $q_{\alpha}(x)$ defined in Eq. (\ref{q_a_def})}
\renewcommand{\theequation}{F.\arabic{equation}}
\setcounter{equation}{0}
We start with explaining the following equality for the channel (\ref{symchan}):
\begin{eqnarray}
\nonumber \displaystyle\sum_{z}\sum_y\displaystyle\frac{\vec{p}_z \cdot \vec{p}^{\alpha}_{y}}{(\vec{p}_z \cdot \vec{p}_{y})^{\alpha-1}}
&=&\displaystyle\sum_z M_z\cdot \displaystyle\frac{\left(M_z+\mu^{\alpha}/\epsilon^{\alpha}-1 \right)}{\left(M_z+\mu/\epsilon-1\right)^{\alpha-1} }+({K}-M_z)M_z^{2-\alpha}\\
&=&\displaystyle\sum_z q_{\alpha}(M_z).
\label{q_a_def222}
\end{eqnarray}
The sum over $y$ is calculated by noticing that the product $[\vec{p}_z \cdot \vec{p}_{y}]$, and, respectively, the product $[\vec{p}_z \cdot \vec{p}^{\alpha}_{y}]$, can have one of two results. If the $1$'s in the binary vector $\vec{p}_z$ overlap only $\epsilon$'s in $\vec{p}_{y}$, we get $[\vec{p}_z \cdot \vec{p}_{y}]=M_z\cdot\epsilon$ and respectively, $[\vec{p}_z \cdot \vec{p}^{\alpha}_{y}]=M_z\cdot\epsilon^{\alpha}$. Otherwise, we get $[\vec{p}_z \cdot \vec{p}_{y}]=\mu+(M_z-1)\cdot\epsilon$ and respectively, $[\vec{p}_z \cdot \vec{p}^{\alpha}_{y}]=\mu^{\alpha}+(M_z-1)\cdot\epsilon^{\alpha}$. It is not hard to see that the second result will occur exactly $M_z$ times in the sum on $y$, and thus the first will occur exactly $K-M_z$ times. The rest is straightforward.
We now prove the concavity of the function $q_{\alpha}(x)$, $ 1 \leq x \leq {K}-M+1$.
%%%%%%%%%%%%%%%%%%%%%%%%%%%%%%%%%%%%%%%%%%%%%%%%%%%%%%%%%%%%%%%%%%%%%%%%%%%
The second derivative of $q_{\alpha}(x)$ is ($c_{\alpha}\equiv \mu^{\alpha}/\epsilon^{\alpha}$):
\begin{eqnarray}
\nonumber q_{\alpha}''(x)&=&
2(2-\alpha)\left(x+\mu/\epsilon-1\right)^{1-\alpha}
-x(2-\alpha)(\alpha-1)\left(x+\mu/\epsilon-1\right)^{-\alpha}\\
%%%%%%%%%%%%%%%%%%%%%%%%%%%%%%%%%%%%%%%%%%%%%%%%%%%%%%%%%%%%%%%%%%%%%%%%%%%%%%%%%%%%%%%%%%%%%%%%%%%%%%
\nonumber &&-2c_{\alpha}(\alpha-1)\left(x+\mu/\epsilon-1\right)^{-\alpha}
+x\cdot c_{\alpha}\cdot\alpha(\alpha-1)\left(x+\mu/\epsilon-1\right)^{-\alpha-1}\\
%%%%%%%%%%%%%%%%%%%%%%%%%%%%%%%%%%%%%%%%%%%%%%%%%%%%%%%%%%%%%%%%%%%%%
\nonumber &&-2(2-\alpha)x^{1-\alpha}-(2-\alpha)({K}-x)(\alpha-1)x^{-\alpha}\\
\nonumber &\leq& 2(2-\alpha)x^{1-\alpha}
-x(2-\alpha)(\alpha-1)\left(x+\mu/\epsilon-1\right)^{-\alpha}\\
%%%%%%%%%%%%%%%%%%%%%%%%%%%%%%%%%%%%%%%%%%%%%%%%%%%%%%%%%%%%%%%%%%%%%%%%%%%%%%%%%%%%%%%%%%%%%%%%%%%%%%
\nonumber &&-2c_{\alpha}(\alpha-1)\left(x+\mu/\epsilon-1\right)^{-\alpha}
+x\cdot c_{\alpha}\cdot\alpha(\alpha-1)\left(x+\mu/\epsilon-1\right)^{-\alpha-1}\\
%%%%%%%%%%%%%%%%%%%%%%%%%%%%%%%%%%%%%%%%%%%%%%%%%%%%%%%%%%%%%%%%%%%%%
&&-2(2-\alpha)x^{1-\alpha}.
\end{eqnarray}
The inequality follows from the assumption $\mu\geq \epsilon$, which leads to $x\leq x+\mu/\epsilon -1$, and from the fact that the last term in the derivative is negative.
Doing some algebraic manipulations, we get:
\begin{eqnarray}
\nonumber q_{\alpha}''(x)
&\leq& -(2-\alpha)(\alpha-1)x\left(x+\mu/\epsilon-1\right)^{-\alpha}\\
%%%%%%%%%%%%%%%%%%%%%%%%%%%%%%%%%%%%%%%%%%%%%%%%%%%%%%%%%%%%%%%%%%%%%%%%%%%%%%%%%%%%%%%%%%%%%%%%%%%%%%
\nonumber &&-2c_{\alpha}(\alpha-1)\left(x+\mu/\epsilon-1\right)^{-\alpha}
+\cdot c_{\alpha}\cdot\alpha(\alpha-1)x\left(x+\mu/\epsilon-1\right)^{-\alpha-1}\\
\nonumber &\propto& -(2-\alpha)x\left(x+\mu/\epsilon-1\right)
-2c_{\alpha}\left(x+\mu/\epsilon-1\right)
+x\cdot c_{\alpha}\cdot\alpha\\
\nonumber &=& -\left[(2-\alpha)x
+2c_{\alpha}\right]\left(x+\mu/\epsilon-1\right)
+x\cdot c_{\alpha}\cdot\alpha\\
\nonumber &\leq&-\left[(2-\alpha)x
+2c_{\alpha}\right]x
+x\cdot c_{\alpha}\cdot\alpha\\
\nonumber &\propto&-(2-\alpha)x
-2c_{\alpha}
+c_{\alpha}\cdot\alpha\\
\nonumber &=&-(2-\alpha)x
-c_{\alpha}(2-\alpha)\\
&<&0.
\end{eqnarray}
Notice that the constant of proportionality is positive in all cases.
We showed that $q_{\alpha}(x)$ is concave for $1<\alpha<2$. Checking the concavity/convexity out of this range can be done numerically, by simple calculation of  $q''_{\alpha}(x)$. 
%\subsection*{Acknowledgments}
%This research is supported by the Israeli Science Foundation (ISF),
%grant no.\ 208/08.

\newpage
\newenvironment{cvlist}[1]{
\begin{list}
{}{\settowidth{\leftmargin}{#1}
\addtolength{\leftmargin}{2em}
\setlength{\rightmargin}{0in}
\setlength{\labelsep}{0in}
\setlength{\labelwidth}{\leftmargin}
\setlength{\parsep}{0in}    }}{
\end{list}}

\newcommand{\cvitem}[1]{\item[#1\hfill]}

\pagestyle{myheadings}
\thispagestyle{empty}
\begin{center}
Avraham Reani - Biography
\end{center}

{\bf Avraham Reani} received the B.Sc. and M.sc degrees in Electrical Engineering, and the B.A degree in Physics, from the Technion, 
Israel Institute of Technology, in 2003, 2010 and 2003, respectively.
During 2003--2006 he was serving as an R\&D officer at the IDF. 
During 2007--2008 he was an engineer at Intel, Haifa, Israel.
He is now a Ph.D student at the Electrical Engineering Department
of the Technion, under the supervision of Prof. Neri Merhav.

\begin{center}
Neri Merhav - Biography
\end{center}

{\bf Neri Merhav} (S'86--M'87--SM'93--F'99) was born 
in Haifa, Israel, on March 16, 1957. He received the
B.Sc., M.Sc., and D.Sc.\ degrees from the Technion, 
Israel Institute of Technology,
in 1982, 1985, and 1988, 
respectively, all in electrical engineering.

From 1988 to 1990 he was with AT\&T Bell Laboratories,
Murray Hill, NJ, USA. 
Since 1990 he has been with the 
Electrical Engineering Department
of the Technion, where
he is now the Irving Shepard Professor.
During 1994--2000 he was also serving 
as a consultant to the Hewlett--Packard 
Laboratories -- Israel (HPL-I). 
His research interests include
information theory, statistical communications,
and statistical signal processing. He is especially
interested in the areas of lossless/lossy source coding
and prediction/filtering, relationships between information
theory and statistics, detection,
estimation, as well as 
in the area of Shannon Theory, including topics in
joint source--channel coding, source/channel simulation,
and coding with side information with applications to
information hiding and watermarking systems.
Another recent research interest concerns the relationships between
Information Theory and statistical physics.

Dr.\ Merhav was a co-recipient of the 1993 Paper Award 
of the IEEE Information Theory Society and he is a
Fellow of the IEEE since 1999. He also received the 
1994 American Technion Society 
Award for Academic Excellence and the 2002
Technion Henry Taub Prize for Excellence in Research.
From 1996 until 1999 he served as an Associate Editor
for Source Coding to the 
{\sc IEEE Transactions on Information Theory}.
He also served as a co--chairman of the Program Committee
of the 2001 IEEE International Symposium on Information Theory.
He is currently on the Editorial Board of {\sc Foundations
and Trends in Communications and Information Theory}.
\end{document}